\documentclass[11pt]{amsart}

\usepackage{hyperref} 
\usepackage{macros}
\usepackage{stmaryrd}
\date{}
\newcommand{\gl}{\mf{gl}}

\begin{document}
\title{Twisted supergravity and its quantization}
\author{Kevin Costello and Si Li}
\thanks{}

\address{Perimeter Institue for Theoretical Physics}
\email{kcostello@perimeterinstitute.ca}

\address{Yau Mathematical Sciences Center, Tsinghua University}
\email{sli@math.tsinghua.edu.cn}
\maketitle

\begin{abstract}
Twisted supergravity is supergravity in a background where the bosonic ghost field takes a non-zero value.  This is the supergravity counterpart of the familiar concept of twisting supersymmetric field theories. 

In this paper, we  give conjectural descriptions of type IIA and IIB supergravity in $10$ dimensions. Our conjectural descriptions are in terms of the closed-string field theories associated to certain topological string theories, and we conjecture that these topological string theories are twists of the physical string theories. 

For type IIB, the results of \cite{CosLi15} show that our candidate twisted supergravity theory admits a unique quantization in perturbation theory. This is despite the fact that the theories, like the original physical theories, are non-renormalizable.   

Although we do not prove our conjectures,  we amass considerable evidence. We find that our candidates for the twisted supergravity theories contain the residual supersymmetry one would expect. We also prove (using heavily a result of Baulieu \cite{Bau10}) the open string  version of our conjecture:  the theory living on a brane in the topological string theory is a twist of the maximally supersymmetric gauge theory living on the brane in the physical string theory.

\end{abstract}
 
\section{Introduction}
Since its introduction by Witten \cite{Wit88a}, the idea of twisting supersymmetric field theories has had a profound impact on quantum field theory and on related areas of mathematics.  Twisted supersymmetric field theories know about certain BPS operators in the original physical theory, and often one can perform exact computations in the twisted theory.  In this way one can obtain many checks on the conjectural dualities of string and $M$-theory.

One important duality is not accessible by this method, however: the AdS/CFT correspondence.  This is because the AdS/CFT correspondence relates a supersymmetric field theory with a supergravity theory, and Witten's construction of twisting does not apply to supergravity.

In this paper we remedy this gap, by proposing a definition of twisted supergravity which we conjecture is AdS dual to twisted supersymmetric field theory. We give conjectural formulations of twisted type IIA and type IIB supergravity in $10$ dimensions.  For example, we conjecture that twisted type IIB supergravity is the Kodaira-Spencer theory \cite{BerCecOog94} (or BCOV theory as extended for arbitrary Calabi-Yau manifolds in \cite{CosLi12}) in $5$ complex dimensions.  

Perhaps surprisingly, type IIB  twisted supergravity theory can be quantized uniquely in perturbation theory, despite the fact that the theory is non-renormalizable. This result was proved in \cite{CosLi15} where we showed that BCOV theory on $\C^5$ admits a canonically defined perturbative quantization.  We believe a similar argument applies to type IIA; we plan to present the details elsewhere.  

These results are the first step in a program to give a precise formulation of the AdS/CFT correspondence in terms of rigorously defined mathematical objects: quantum twisted supergravity and quantum twisted supersymmetric gauge theories.  We will explore the AdS/CFT correspondence in detail in  subsequent publications.  The primary aim of this article is to define and motivate the concept of twisted supergravity, to state our conjectural description for various twisted supergravity theories, and to provide evidence for these conjectures.

To be able to formulate the twisted AdS/CFT correspondence, we need to understand not just twisted supergravity, but also how branes appear in twisted supergravity.  In this paper we describe how $D$-branes behave in our twists of type IIA and type IIB. We find that the theory living on a $D$-brane in twisted type IIB supergravity can be described as holomorphic Chern-Simons theory on a certain supermanifold, and that there is a similar description for $D$-branes in type IIA. We prove (using a result of Baulieu) that the theory living on a  $D$-brane in twisted supergravity in $10$ dimensions is a twist of the theory living on the brane in the physical theory.  

We hope that our formulation of twisted string theory and supergravity will make aspects of the subject accessible to mathematicians.  Indeed, we conjecture that twists of type IIA and IIB superstring theory can be described in terms of topological string theories which are mixtures of $A$- and $B$-models. 


\section{Twisted supergravity}
In this section, we will introduce the concept of twisted supergravity and discuss its relation with twisted supersymmetric field theories.

\subsection{Reminders on twisting supersymmetric field theories}
Before describing twisted supergravity, let us first recall what a twisted supersymmetric field theory is.  Suppose we have a supersymmetric field theory in dimension $d$, which is acted on by the Lorentz group $\op{Spin}(d-1,1)$ and has an $R$-symmetry group $G_R$. Both $\op{Spin}(d-1,1)$ and $G_R$ act on the supersymmetries of the theory.   As Witten described it, twisting the theory amounts to performing the following three steps.
\begin{enumerate} 
 \item First, we choose a homomorphism $\phi : \op{Spin}(d-1,1) \to G_R$.   This gives rise to a new action of $\op{Spin}(d-1,1)$ on the fields of the theory, using the homomorphism
 $$
(1,\phi) : \op{Spin}(d-1,1) \to \op{Spin}(d-1,1) \times G_R 
 $$
and the action of $\op{Spin}(d-1,1) \times G_R$ on the fields. Using these new action of $\op{Spin}(d-1,1)$ means that we have changed the spin of the fields, and also of the supersymmetries. 
\item Next, we choose a supercharge $Q$ which is invariant under the new action of  $\op{Spin}(d-1,1)$, and which has $Q^2 = 0$. 
\item Finally, we add $Q$ to the BRST differential $\d_{BRST}$ of the theory, and treat the differential $\d_{BRST} + Q$ as the BRST differential of the twisted theory. For example, the space of physical operators in the twisted theory is the cohomology of the space of operators of the original theory with respect to $\d_{BRST} + Q$. Typically, this cohomology is much smaller than the cohomology with respect to just $\d_{BRST}$. 
\end{enumerate}

The criterion that $Q$ must be invariant under a copy of $\op{Spin}(d-1,1)$ implies that the image of the operation of bracketing with $Q$ in the supersymmetry algebra is an $\op{Spin}(d-1,1)$-invariant subspace of $\R^{d-1,1}$, and so must be the whole space (otherwise $Q$ would be central).  This tells us that translation in every direction in space-time is $Q$-exact.  This implies that the Hamiltonian and momentum operators in the twisted theory are zero, and that the correlation functions of local operators are independent of position. These are hallmarks of having a topological theory.

If one is willing to consider non Lorentz-invariant theories, there is no need for step 1.  We can just choose any supercharge $Q$ of square zero, and then add $Q$ to the BRST differential.  These leads to a more general class of twists than those originally considered by Witten. These twists may not be topological.

For example, working in Euclidean signature on $\R^{2d} = \C^d$,  there is typically\footnote{There is such a $Q$ whenever there are at least $4$ supercharges} a unique $Q$ (up to rotation by $R$-symmetry) which is $SU(d)$-invariant, and for which the translations $\dpa{\zbar_i}$ in anti-holomorphic directions are $Q$-exact.  This is a holomorphic twist.   The holomorphic twist is present even when there is no topological twist, and will contain more information about the physical theory than any topological twist will.  

We mention this because the twists of supergravity that we will be interested in will be of this more general kind.

\subsection{Twisted supergravity}
In supergravity, the local supersymmetry algebra is a gauge symmetry of the theory.  

In ordinary Einstein gravity, the fundamental field is a metric, and the diffeomorphisms play the role of gauge symmetries. The functional integral (morally speaking) is an integral over a space of metrics modulo gauge.  In the BV formalism, which is how we like to approach field theory, one models the quotient of the space of metrics by diffeomorphisms by introducing ghost fields corresponding to vector fields on the manifold.  These ghost fields are fermionic (i.e. anti-commuting); they have cohomological degree $-1$.

In supergravity, there are more fields and more symmetries.  The precise field content will not matter right now; however, it is important to understand the extra symmetries.  As well as diffeomorphism symmetries, there are extra \emph{fermionic} symmetries called local supersymmetries.  These are sections of a bundle of spinors (which bundle of spinors appears depends on the theory considered).  Like in ordinary gravity and in gauge theory, one should introduce ghost fields to model taking the quotient of the space of fields by these symmetries.  

The ghost field corresponding to an ordinary bosonic symmetry is fermionic. Similarly, the ghost corresponding to a fermionic symmetry is \emph{bosonic}.  For a supergravity theory on $\R^n$, where the local supersymmetries take value in some spin representation $S$ of $\op{Spin}(n,1)$, then the bosonic ghost fields are elements
$$q \in \cinfty(\R^{n}, S),$$
that is, $q$ is an $S$-valued smooth functions on $\R^{n}$. 

The main idea is that: 

\fbox{\parbox[c]{14cm}{Twisted supergravity is supergravity in a background where the bosonic ghost field takes some non-zero value}}

  To understand this better, however, we need to understand the equations of motion satisfied by this bosonic ghost field.  Suppose, for simplicity, we consider a supergravity configuration on $\R^{n}$ where the only fields that take a non-zero value are the metric $g$ and the bosonic ghost field $q$.  Then, this configuration satisfies the equations of motion if and only if 
\begin{enumerate}
\item $g$ satisfies the usual equations of motion of supergravity, that is, it is Ricci flat. 
\item The spinor $q$ is covariant constant for the metric given by $g$.
\item The spinor $q$ satisfies $\Gamma(q , q) = 0$, where $\Gamma : S \otimes S \to \R^n$ is the map defining the Lie bracket on the supersymmetry algebra. (Thus, the components of $\Gamma$ are $\Gamma$-matrices). 
\end{enumerate}
 
The fact that $g$ must be Ricci flat is immediate.  To see the second equation, note that the supergravity action functional, in the BV formalism, will have a term like 
$$\int \psi^\ast \nabla_g q $$ 
where $\psi^\ast$ is the anti-field to the gravitino, and $\nabla_g$ refers to the covariant derivative of the ghost field $q$ with respect to the metric $g$.  This term reflects the gauge symmetry of the gravitino of the form $\psi \mapsto \psi + \delta \nabla_g q$. Variation of this term with respect to $\psi^\ast$ will tell us that $\nabla_g q = 0$. 

There will also be a term of the form 
$$\int V^\ast  \Gamma(q,q)$$
where $V^\ast$ is the anti-field to the ghost for the  diffeomorphism gauge symmetry, and $\Gamma(q,q)$ is viewed as a vector field on $\R^{n}$ and thus as a ghost for the diffeomorphism symmetry. Varying this term with respect to $V^\ast$ leads to the equation that $\Gamma(q,q) = 0$.

\subsection{}
Recall that when we discuss quantum field theory on $\R^n$,  to define the functional integral, we need to not only specify the Lagrangian, but we also need to specify the behaviour of the fields at infinity.  Typically one requires that the fields converge towards some fixed constant solution to the equations of motion at infinity.  This is the choice of a vacuum of the theory. 

We can do the same thing in supergravity. We can consider the theory on flat space where at $\infty$ we require the fields to converge towards some fixed vacuum solution. 
\begin{definition}
Twisted supergravity on $\R^{n}$ is supergravity in the vacuum where the metric is the standard flat metric, and the bosonic ghost field take some non-zero value $q$ with $\Gamma(q,q) = 0$. 
\end{definition}
Note that this vacuum will not be Lorentz invariant, only translation invariant. If one considers supergravity in perturbation theory, the fields will be small fluctuations of this vacuum. 

Of course, one can consider twisted supergravity on other manifolds than $\R^{n}$.  
\begin{definition}
If $M$ is a manifold of dimension $n$, then a twisted background for supergravity on $M$ is a solution to the equations of motion where the bosonic ghost takes some non-zero value. 
\end{definition}
Such a solution to the equations of motion will consist of a metric on $M$, together with  some other bosonic fields (such as flux fields), which satisfy the usual equations of motion.  These fields must admit a generalized Killing spinor (i.e.\ be a BPS solution), and the value of the bosonic ghost is some generalized Killing spinor of square zero.

If we work in perturbation theory, then it makes sense to consider twisted supergravity on a compact manifold. This is just supergravity in perturbation theory around a twisted background. Non-perturbatively, however, this doesn't really make sense, as one must perform the path integral over all field  configurations.  One can, however, imagine finding an approximation to the full non-perturbative supergravity path integral by doing perturbation theory around each solution to the equations of motion and then performing a finite-dimensional integral over the space of solutions to the equations of motion.  If one does this, one sees that twisted supergravity backgrounds will contribute to the full supergravity path integral.  

\subsection{Analogy with the Higgs and Coulomb branches of gauge theories}
In this section we will explain an analogy that supergravity in a twisted vacuum is similar to a gauge theory on the Coulomb branch.  

Suppose we have a gauge theory on $\R^n$ with gauge group $G$, with some matter fields. The precise field content is not important for this discussion, but let us assume that the matter fields include a scalar field $\phi$ living in a representation $R$ of the group $G$.  

Suppose we are in a vacuum of the theory where the scalar field $\phi$ has some non-zero expectation value. Then, the theory is said to be on the Higgs branch. We should think of the field $\phi$ as an excitation of a constant field $\phi_0$ valued in the representation $R$ of $G$.  In this situation, the gauge symmetry of the theory is broken to the subgroup $\op{Stab}(\phi_0) \subset G$ of elements which fix the element $\phi_0 \in R$.  Often, this subgroup is trivial.

At the origin of the Higgs branch, the element $\phi_0 \in R$ takes value $0$. In this case, the gauge symmetry of the theory is not broken at all, and another branch of the moduli of vacua breaks off: the Coulomb branch. On the Coulomb branch, expectation values of operators constructed from the gauge field have some non-zero value.  

Something similar happens in supergravity. Suppose we have a supergravity theory in $d$ dimensions, coupled to some supersymmetric field theory.  Let us suppose that this supersymmetric theory has some (possibly disconnected) space of vacua, which contains a smaller space of supersymmetric vacua. Let's fist consider ``ordinary'' vacua of this coupled theory, where the bosonic ghost fields have zero expectation value.  A generic such vacuum will have no supersymmetry, just like a generic vacuum on the Higgs branch of a suitable gauge theory will have no gauge symmetry.   At supersymmetric vacua an extra branch of the moduli space of vacua breaks off, on which the bosonic ghosts have non-zero expectation value. This is completely analogous to the way the moduli of  vacua of a gauge theory can acquire a Coulomb branch connecting to the Higgs branch at vacua with non-trivial gauge symmetry.

\subsection{The relation between twisted supergravity and twisted supersymmetric field theories}
Any supersymmetric field theory in dimension $n$ with   $\mscr{N}$-extended supersymmetry can be coupled to the corresponding $\mscr{N}$-extended supergravity theory. In particular, if we treat the supergravity theory as a classical theory, we can consider the supersymmetric field theory in any supergravity background.  We are interested in what happens to a supersymmetric field theory placed in a twisted supergravity background, i.e. a background where the bosonic ghost field has non-zero value. 

To be concrete, let's consider our supersymmetric theory on a manifold $M$ in a supergravity background where the only non-zero fields are the metric and the bosonic ghost field $q$. We require that the bosonic ghost $q$ takes value $q = \Psi$, where $\Psi$  is a covariant-constant spinor-field $\Psi$ with $\Gamma(\Psi,\Psi) = 0$.   Let $Q_{\Psi}$ be the corresponding supersymmetry acting on the super-symmetric field theory.
\begin{lemma}
The supersymmetric field theory in this supergravity background is the same as the supersymmetric theory  twisted by the supercharge $Q_{\Psi}$.  (Recall that by twisting we mean that we add $Q_{\Psi}$ to the BRST operator $Q_{BRST}$). 
\end{lemma}
\begin{proof}
Let us consider the BRST operator of the theory coupling the supersymmetric field theory to the supergravity theory (where we consider both theories, for now, at the classical level).  In this coupled theory, we have gauged the action of local supersymmetry on both the supergravity theory and the supersymmetric field theory.  

It follows that one of the terms in this BRST operator must be the Chevalley-Eilenberg differential for the action of local supersymmetry on fields.  Let us write down this term in the BRST operator explicitly.  To fix notation, let us pick a basis of the (infinite-dimensional) space of local supersymmetries, say $\Psi_i$, and let $f^i$ be the corresponding dual basis of the dual vector space. We will view the functions $f^i$ as parametrizing the bosonic ghost-field. Then, the term in the BRST operator which comes from this Chevalley-Eilenberg differential, applied to a local operator $\Oo$ of the supersymmetric field theory,  is  of the form
$$
\Oo \mapsto \sum f^i Q_{\Psi_i}\Oo
$$ 
Here by $Q_{\Psi_i}$ we mean the action of supersymmetry corresponding to $\Psi_i$ on the operators of the gauge theory.

If the bosonic ghost field $q$ is  given a particular value $q =\Psi$, we find that we have added $\sum f^i(q) Q_{\Psi_i} =Q_{\Psi}$ to the BRST differential of the field theory. 
\end{proof}

Although we will not focus so much in this paper on the other aspect of twisting supersymmetric field theories, whereby one changes the spin if the fields using a homomorphism from $\op{Spin}(d)$ to the $R$-symmetry group, this can also be given a supergravity interpretation \cite{KarRoc88}.  For this to work, we need to use a model of supergravity in which the $R$ symmetry group is gauged.  Then, given a homomorphism $\rho: \op{Spin}(d) \to G_R$, one can construct a principal $G_R$ bundle with connection on any Riemannian $d$-manifold as the bundle associated to the frame bundle with its Levi-Civita connection, using the homomorphism $\rho$.  Since the $R$ symmetry group is gauged,  a Riemannian manifold equipped with this $G_R$ bundle defines a supergravity background.  Putting a supersymmetric field theory in a background of this form has the effect of changing the spin of the fields.  

\section{Summary of the rest of the paper}
The rest of the paper contains conjectural descriptions of certain twists of type IIA and type IIB supergravity theories, together with evidence for these conjectures.  We start by reviewing some features of BCOV theory, which we will need to formulate our conjectural descriptions of twisted supergravity theories.  Then, we discuss twisted type IIB supergravity. We conjecture that twisted type IIB supergravity on $\R^{10}$ is BCOV theory on $\C^5$. We further conjecture that the twist of the type IIB superstring is the $B$-model topological string theory on $\C^5$. 

For formal reasons, one expects to find the $Q$-cohomology of the $(2,0)$ supersymmetry algebra of type $IIB$ supergravity appearing in the fields of the twisted supergravity theory.  Here $Q$ is the supercharge which we use to twist.  We calculate this $Q$ cohomology algebra and find that it does indeed appear in the fields of BCOV theory on $\C^5$.  There are some subtle points here: the supersymmetry algebra of type IIB that we analyze is not simply the $(2,0)$ supersymmetry algebra in $10$ dimensions, but rather its central extension corresponding to the fundamental string.  We find that, for the commutator in the residual algebra of supersymmetries to hold in BCOV theory, it is necessary to take account of this central extension.

Then, we analyze the relation between $D$-branes in type IIB superstring theory and branes in the topological $B$-model on $\C^5$. We find that the holomorphic twist of the theory living on a $D_{2k-1}$ brane (i.e.\ of the maximally supersymmetric gauge theory in $2k$ dimensions) is the theory living on a topological $B$-brane on $\C^k \subset \C^5$, which is a version of holomorphic Chern-Simons.  Baulieu \cite{Bau10} proved this in the case that $k = 5$, and we derive the general case as a consquence of his result.  

As a further consistency check, we calculate that the residual supersymmetry that acts on the holomorphic twist of the theory on a $D_{2k-1}$ brane matches with part of the $Q$-cohomology of the $10$-dimensional $(2,0)$ supersymmetry algebra that appears in BCOV theory.

\subsection{Quantization}
The results of \cite{CosLi15} to show that BCOV theory (and so, conjecturally, twisted type IIB supergravity) can be perturbatively quantized on $\C^5$.  To have a self-contained story, we will explain a little about these results from the point of view expounded here.

The main theorem of \cite{CosLi15} concerns BCOV theory on $\C^5$ coupled to holomorphic Chern-Simons theory for the group $\mf{gl}(N \mid N)$.  The statement is that the coupled open-closed theory can be quantized in  a unique way when we work in a uniform way  in $N$. The proof has two parts: first, we verify that a certain one-loop anomaly vanishes. This is somewhat analogous to the Green-Schwarz anomaly cancellation.  Then, we construct the theory at all loops using a cohomological argument, whereby possible counter-terms from the open-string sector cancel precisely with those from the closed string sector.

Holomorphic Chern-Simons on $\C^5$ is a twist of $10$-dimensional supersymmetric gauge theor.  This supersymmetric gauge theory with gauge group $GL(N)$ is the theory living on a stack of $N$ $D9$ branes in type IIB.  It has recently been argued \cite{DijHeiJefVaf16}, that one can realize the theory with gauge group the super-group $GL(N \mid N)$ from a stack of $N$ $D9$ and $N$ anti-$D9$ branes. 

Thus, the main result of \cite{CosLi15} can be interpreted as saying that, after twisting, the theory coupling type IIB supergravity with $N$ $D9$ and $N$ anti-$D9$ branes has a unique quantization in perturbation theory.  

It is natural to ask whether a similar statement holds before twisting, and in particular whether the one-loop anomaly cancellation result of \cite{CosLi15} can be seen in the physical theory.   This would be an interesting type IIB analog of Green-Schwarz's classic result.  
\subsection{Type IIA}
We perform a similar analysis for type IIA. We consider an $SU(4)$ invariant twist of type IIA supergravity and superstring theory.  We conjecture that the twist of the string theory is a topological string theory on $\R^2 \times \C^4$ which is the topological $A$-model on $\R^2$ and the topological $B$-model on $\C^4$.  Similarly, we conjecture that the twist of type IIA supergravity theory is the closed string field theory constructed from this topological string (which we write down explicitly). 

We analyze the $Q$-cohomology of the $10$-dimensional $(1,1)$ supersymmetry algebra of type IIA. We find this supersymmetry algebra present in our conjectural description of twisted type IIA supergravity.  As in the case of type IIB, for this to work we need to use the central extension of the $(1,1)$ supersymmetry algebra in $10$ dimensions corresponding to the fundamental string. 

We then analyze the $D$-branes in type IIA in a manner parallel to our analysis of branes in type IIB.  A brane in the topological string theory which we conjecture is the twist of type IIA lives on a submanifold of $\R^2 \times \C^4$ of the form $\R \times \C^k$ where $\R \subset \R^2$ and $\C^k \subset \C^4$ are linear subspaces.  We prove that the theory living on this brane is a twist of the theory living on the $D_{2k}$ brane in the physical string theory.  We also show that the residual supersymmetries present in our twisted type IIA supergravity which preserve the brane are precisely the residual supersymmetries present in the twist of the $D_{2k}$-brane gauge theory.

One expects twists of  type IIA and type IIB superstring theory to be related by $T$-duality. We verify that our twist of type IIA on $\R \times S^1 \times \C^4$ is $T$-dual to our twist of type IIB on $\C^\times \times \C^4$.  We also verify that type IIA and type IIB supergravity theories become the same when reduced to $9$ dimensions.

\section{BCOV theory}
Our conjectural descriptions of twisted supergravity theories are mostly in terms of BCOV theory (also known as Kodaira-Spenser theory), whose definition (as formulated in \cite{CosLi12} generalizing that in \cite{BerCecOog94}) we briefly recall. We refer to \cite{CosLi12, CosLi15} for detailed discussion. 

 If $X$ is a Calabi-Yau of dimension $d$, let 
$$
\PV^{i,j}(X) = \Omega^{0,j}(X,\wedge^j TX)
$$ 
be the space of poly-vector fields of type $(i,j)$ on $X$. Contraction with the holomorphic volume form yields an isomorphism $\PV^{i,j}(X) \iso \Omega^{d-i,j}(X)$ and so to operators $\dbar, \partial$ on $\PV^{\ast,\ast}(X)$ which correspond via this isomorphism to the operators of the same name on $\Omega^{\ast,\ast}(X)$. Note that $\partial$ maps $\PV^{i,j}$ to $\PV^{i-1,j}$, which is the divergence operator with respect to the holomorphic volume form.  Further, $\oplus \PV^{i,j}(X)$ is a graded commutative algebra, with product being wedge product in the exterior algebra of the tangent bundle. When $X$ is non-compact, we will denote by 
$$
   \PV_c^{i,j}(X)\subset \PV^{i,j}(X), \quad \Omega^{i,j}_c(X)\subset \Omega^{i,j}(X)
$$
the subspace consisting of compactly supported elements. 

There is an integration map $\int: \PV^{d,d}_c(X) \to \C$ which comes from the isomorphisms 
$$\PV^{d,d}(X) \iso \Omega^{0,d}_c(X) \iso \Omega^{d,d}_c(X)$$ where the first isomorphism is given by contraction with the holomorphic volume form and the second by wedging with the holomorphic volume form.  We will extend this integration map to all of $\PV^{\ast,\ast}(X)$ by setting it to be zero on $\PV^{i,j}(X)$ if $(i,j) \neq (d,d)$.  

In the original formulation \cite{BerCecOog94}, the fields of Kodaira-Spencer theory is the subspace
$$
\op{Ker} \partial \subset \oplus \PV^{i,j}(X)[2].
$$
Here $[2]$ is the shift of cohomology degree by $2$ such that $\PV^{1,1}$ is of degree $0$. The action functional is the non-local functional
$$
\tfrac{1}{2} \int \alpha \dbar \partial^{-1} \alpha + \tfrac{1}{6} \int \alpha^3.
$$

The presence of a non-local action functional is a little unsatisfactory, as is the fact that the fields are subject to a constraint which does not come from the equations of motion. In our formulation \cite{CosLi12}, we present a different approach to the theory which does not have these defects.

We introduce a larger space of fields which form a complex locally resolving the space $\op{Ker} \partial$.  Our space of fields is
$$
\PV^{\ast,\ast}(X)\left\llbracket t\right\rrbracket[2]
$$
where $t$ is a formal variable of cohomological degree $2$.  This space of fields has a differential, corresponding to the linearized BRST operator, which is $\dbar + t \partial$.  

Rather than presenting an action functional with a non-local quadratic term, we will present the theory as a degenerate theory in the BV formalism. In the BV formalism, a field theory is specified by a dg manifold with an odd symplectic structure. Locally, one reconstructs the action functional from this data by representing the differential as the Hamiltonian vector field associated to a function; this function is the action functional. 

 A degenerate theory is specified by  a dg manifold with an odd Poisson structure.  Since not every Poisson vector field is Hamiltonian, a degenerate theory does not necessarily have an action functional.   We will find that our theory can be specified by a dg odd Poisson manifold where the linear term in the differential is not Hamiltonian, but the non-linear terms are.  We can alternatively describe the theory as being given by a linear odd dg Poisson manifold, with constant-coefficient Poisson tensor, together with a functional describing the interactions. This functional satisfies the classical master equation. Adding the Poisson bracket with this functional to the differential deforms the linear Poisson manifold to a non-linear one.

For us, the dg manifold describing free BCOV theory is the complex $\PV^{\ast,\ast}(X)\left\llbracket t\right\rrbracket[2]$ with differential $\dbar + t \partial$. The BV Poisson tensor is
\begin{align*} 
 (\partial \otimes 1) \delta_{Diag} \in \bigoplus_{\substack{i_1 + i_2 = d+1 \\j_1 + j_2 = d}} \Omega^{i_1,j_1}(X) \otimes \Omega^{i_2,j_2}(X)
 \end{align*}  
 or in terms of polyvector fields
 \begin{align*}
 (\partial \otimes 1) \delta_{Diag}  \in \bigoplus_{\substack{i_1 + i_2 = d-1 \\j_1 + j_2 = d}}t^0 \PV^{i_1,j_1}(X) \otimes t^0 \PV^{i_2,j_2}(X) 
\end{align*}
where $\delta_{Diag}$ refers to the $\delta$-current on the diagonal of $X \times X$, and we have used the isomorphism between forms and polyvector fields referred to above.  We are abusing notation in our use of the symbol $\otimes$ above: really we mean an appropriate completion which describes polyvector fields on $X \times X$ with distributional coefficients. The distribution is understood in the sense that for any $\alpha\in \PV^{i,j}(X), \beta \in \PV^{d+1-i,d-j}(X)$, 
$$
(\partial \otimes 1) \delta_{Diag}: (\alpha\otimes \beta)\to \int_X \alpha \partial \beta.
$$

In this way, $\PV(X)\left\llbracket t\right\rrbracket[2]$ acquires the structure of linear dg Poisson manifold where the Poisson tensor is odd. Note that the Poisson tensor lives in the subspace of fields which have no powers of $t$. 

For a theory which is degenerate in the BV formalism, and with a constant-coefficient odd Poisson tensor, the propagator $P$ with an infrared cutoff is constructed so that it provides a homotopy between the odd Poisson tensor and a different odd Poisson tensor which is smooth. That is, the propagator is an element 
$$
P \in \PV(X)[[t]] \what{\otimes} \PV(X)[[t]] 
$$ 
satisfying
$$
(\dbar + t \partial) P = (\partial \otimes 1) \delta_{Diag} + \text{ something smooth (the regularized Poisson kernel) }. 
$$
Here $\dbar + t \partial$ indicates the sum of the differentials from each copy of $\PV(X)[[t]]$, and $\what{\otimes}$ is an appropriate completed tensor product.  The propagator can always be chosen so that it has no dependence on $t$ in either factor.

Explicitly, on flat space, it is (up to constant factors) 
\begin{align*} 
P =&  \dbar^\ast_z \partial_z\frac{ (\partial_{z_1} - \partial_{w_1}) \dots (\partial_{z_n} -\partial_{w_n})(\d \zbar_1 -  \d \br{w}_1) \dots (\d \zbar_n - \d \br{w}_n)  )}{ \norm{z-w}^{2n-2}  } \\
=& \sum (-1)^i (-1)^{j+n} (\partial_{z_1} - \partial_{w_1}) \dots\what{(\partial_{z_i} - \partial_{w_i}) }\dots  (\partial_{z_n} -\partial_{w_n})\\
&\times (\d \zbar_1 -  \d \br{w}_1) \dots\what{(\d \zbar_j - \d \br{w}_j)} (\d \zbar_n - \d \br{w}_n)  )   \dpa{z_i} \dpa{z_j} \norm{z-w}^{2-2n}.   
\end{align*}
From this explicit expression one sees that the propagator has no dependence on the $t$ variables.  Although degenerate BV theories may seem exotic, one works with them by applying the usual Feynman rules built from this propagator and the interaction we will write down shortly. 

One might be tempted to throw away all the fields which involve $t$, because these fields do not propagate. It is not possible to do this in a consistent way, because the linearized BRST operator $\dbar + t \partial$ involves $t$.  What one can do, however, is restrict  attention to those fields which live in $t^k \PV^{i,\ast}(X)$ where $i+k \le d-1$.  Fields of this form are closed under the operator $\dbar + t \partial$, and contain all propagating fields.   We will refer to the version of BCOV theory which only includes these fields as \emph{minimal} BCOV theory. In our conjectural relationship between supergravity/superstring theory and BCOV theory, we expect that minimal BCOV theory should relate to supergravity, whereas the remaining fields of BCOV theory should correspond to additional closed string states that do not appear in supergravity.

In \cite{CosLi15} we described a classical interaction for our formulation of BCOV theory, which we now recall. Define functionals
$$
I_n : \PV^{\ast,\ast}_c(X)\left\llbracket t\right\rrbracket[2] \to \C
$$ 
as follows. If $\alpha \in \PV^{\ast,\ast}_c(X)\left\llbracket t\right\rrbracket[2]$, let $\alpha_k$ denote the coefficient of $t^k$. Then, we set
$$
I_n(\alpha) = \sum_{k_1, \dots, k_n \text{ with } \sum k_i = n-3}\frac{(n-3)!}{k_1 ! \dots k_n !} \int \alpha_{k_1} \wedge \dots \wedge \alpha_{k_n}. 
$$
We then define the interaction $I$ by saying that
$$
I(\alpha) = \sum_{n \ge 3} \tfrac{1}{n!} I_n(\alpha). 
$$
One can check \cite{CosLi15} that $I(\alpha)$ satisfies the classical master equation
$$
QI +{1\over 2}\{I, I\}=0.
$$ 
Here the differential $Q=\dbar + t \partial$ on the fields induces via duality a derivation on the functionals, which is the meaning of $QI$. $\{,\}$ is the odd Poisson bracket induced by the above odd Poisson tensor, which is well-defined for local functionals. 

This classical action can be quantized with the help of the technique of renormalization. The precise formulation is described in \cite{CosLi12} where $I$ is quantum corrected to satisfy an effective version of the quantum master equation. This can be viewed as a B-twisted topological version of the closed string field theory in the sense used by Sen and Zwiebach \cite{Zwi93,SenZwi94}.

The open-closed string field theory (in the sense of \cite{Zwi98}) for the topological $B$-model is obtained by coupling BCOV theory to holomorphic Chern-Simons theory.  Let us assume that $d = \op{dim} X$ is odd. Then, the fields of holomorphic Chern-Simons (in the BV formalism) are $\Omega^{0,\ast}(X)\otimes \g[1]$, where $\g$ is the Lie algebra of the gauge group.  

The main theorem of \cite{CosLi15} is the following.
\begin{theorem*}
On $\C^d$ for $d$ odd, there is a unique perturbative quantization of open-closed BCOV theory, where in the open sector the gauge Lie algebra is the super Lie algebra $\gl(N \mid N)$, and we require the quantization to be compatible with inclusions $\gl(N \mid N) \into \gl(N + k \mid N + k)$.
\end{theorem*}
The main ingredient in the proof is a remarkable cancellation between the cohomology  groups describing possible local counter terms between the open sector and the closed sector.   
Since our conjectural descriptions of twisted supergravity theories are all in terms of variants of BCOV theory, this theorem will allow us to produce quantizations. 

There is an important caveat which we should mention about this theorem. The construction of open-closed BCOV theory starts with free BCOV theory and interacting classical holomorphic Chern-Simons, and generates the full quantum theory (including interacting BCOV) by an obstruction theory argument.  In particular, a classical interaction for  BCOV is generated by this procedure: we refer to it as the \emph{dynamically-generated} classical interaction. In \cite{CosLi15} we conjecture, but do not prove, that the dynamically-generated classical interaction is the same as the one we wrote down above. We do, however, verify that they have the same cubic term $I_3(\alpha)$ (using the notation given above). 

 The quartic and higher terms in the action functional all involve descendent fields, so the distiction between the action functional we wrote down above and the dynamically-generated classical interaction is not important for most purposes.  In any case, in what follows, we will generally assume this conjecture when we make statements linking BCOV theory to supergravity. 
\section{Twisted type IIB supergravity}

Consider type IIB supergravity on a Calabi-Yau $5$-fold $X$, in the background where the metric is the Calabi-Yau metric and where other bosonic fields are set to zero. We would like to twist type IIB supergravity on $X$ by setting a bosonic ghost field to a non-zero value.

The local supersymmetries of type IIB supergravity are sections of a rank 32 bundle, which decomposes as a direct sum of two copies of the 16 dimensional bundle of spinors of positive chirality.  Thus, the bosonic ghost fields have $32$ components.  

Since $X$ is Calabi-Yau, there are two covariant constant spinors on $X$, one of each chirality.  Therefore, type IIB supergravity on $X$ has two independent Killing spinors $(Q,0)$ and $(0,Q)$, one in each copy of the 16 dimensional positive chirality spin bundle. We can thus consider a twist of type IIB supergravity on $X$ by setting the bosonic ghost field to have value $(Q,0)$. 
\begin{conjecture}
This twist of type IIB supergravity on a Calabi-Yau $5$-fold $X$ is equivalent to BCOV theory on $X$.
\end{conjecture}
\begin{remark}
More precisely, we expect to find \emph{minimal} BCOV theory, as discussed above, which is the sector of BCOV theory where we remove as many as possible of the non-propagating descendent fields.  The remaining fields of BCOV theory should presumably arise as states in type IIB string theory that do not contribute to supergravity.  We will not, in general, discuss the distinction between the two variants of BCOV theory in future.

\end{remark}
In a similar way, if we work on flat space $X = \C^d$, we find that the space of $SU(5)$-invariant killing spinors is two dimensional, spanned by elements $(Q,0)$ and $(0,Q)$.  We can then twist by setting the ghost field to take value $(Q,0)$. The (perturbative\footnote{If we go beyond the perturbative supergravity approximation, for instance by considering the type IIB superstring, then the $R$-symmetry group is $SL(2,\Z)$ which does not act transitively}) $R$-symmetry Lie algebra $\mf{sl}(2,\C)$ acts transitively on the space of $SU(5)$-invariant killing spinors, so it doesn't matter which one we choose. 
\begin{conjecture}
The twist of type IIB supergravity on flat space $\R^{10}$ is minimal BCOV theory on $\C^5$.  
\end{conjecture}
If one believes this conjecture, then the theorem we stated above concerning quantization of BCOV theory implies that twisted type IIB supergravity admits a canonically-defined perturbative quantization.
 
We can also twist type IIB supergravity on $AdS_5 \times S^5$. Recall that the supersymmetries preserving this background form the Lie algebra $\mf{psl}(4 \mid 4)$. Let us view this as a subalgebra of the endomorphisms of a supervector space $\C^{4 \mid 4}$, where the $R$-symmetry $\mf{sl}(4)$ rotates the fermionic directions and the conformal symmetry $\mf{sl}(4)$ rotates the bosonic directions. Up to conjugation, there is a unique fermionic matrix $Q$ in $\mf{psl}(4 \mid 4)$ of rank $(0 \mid 1)$ (meaning that the image of this matrix is of dimension $(0 \mid 1)$.  Such a matrix is necessarily square zero.  We can thus twist type IIB supergravity on $AdS_5 \times S^5$ by setting the bosonic ghost field to take value $Q$. We have a conjectural description of this twist.   

Among the fields of BCOV theory is a closed $5$-form $F \in \Omega^{3,2}$, which we can also think of as a polyvector field of type $(2,2)$. Conjecturally, $F$ corresponds to some of the components of the Ramond-Ramond $5$-form in IIB theory.   In the AdS background, this RR $5$-form has a non-zero value. 
\begin{conjecture}
The twist of the AdS background of type IIB supergravity on 
$$
AdS_5 \times S_5 \simeq (\R^{10} \setminus \R^4,g)
$$
(where $g$ is a particular metric that blows up along the $\R^4$) is BCOV theory on $\C^5 \setminus \C^2$ where the $5$-form field $F \in \Omega^{3,2}$ takes value
$$
F =N \frac{3}{4 i \pi^3} \d z_1 \d z_2 \d z_3 r^{-6} \left( \zbar_1 \d \zbar_2 \d \zbar_3 - \zbar_2 \d \zbar_1 \d \zbar_3 + \zbar_3 \d \zbar_1 \d \zbar_2  \right). 
$$
Here  $\C^2=\{z_1=z_2=z_3=0\}$ and $r=\sqrt{|z_1|^2+|z_2|^2+|z_3|^2}$ is the radius along the direction normal to $\C^2$. 
\end{conjecture}
An abstract characterization of $F$ is the unique $5$-form on $\C^5 \setminus \C^2$ which extends to a $5$-form with tempered\footnote{Tempered means, roughly, polynomial decay at $\infty$} distributional coefficients on $\C^5$ such that 
$$
\dbar F =N \delta_{\C^2}
$$
where $\delta_{\C^2}$ is the de Rham current representing the fundamental class of $\C^2$. 

Thus, this conjecture states that the only part of the $\op{AdS} \times S^5$ that contributes to the twisted background is the $5$-form field.  

\subsection{Strings}
It is natural to guess that a stronger statement than the one we have made holds.
\begin{conjecture}
The twist of type IIB perturbative superstring theory is the topological $B$-model on $\C^5$.
\end{conjecture}
The twist of the string theory is defined in a similar way to the twist of the supergravity theory. The bosonic ghosts that we have seen in supergravity arise, in string theory, as bosonic Ramond-Ramond states.  Twisted string theory (on flat space) is then string theory in the RR background given by a covariant constant bosonic ghost. 

\subsection{Branes}
The topological $B$-model admits branes which live on holomorphic submanifolds. We conjecture that the image of a $D_{2k-1}$-brane in type IIB in our twisted theory are the $B$-branes in the topological $B$-model living on the submanifolds $\C^k \subset \C^5$.   As evidence for this, we use a result of Baulieu to prove the following. In the topological $B$-model, from each brane we can construct an open-string field theory which is a kind of holomorphic Chern-Simons theory.   
\begin{proposition*}
The holomorphic twist of the theory living on a $D_{2k-1}$ brane is the holomorphic Chern-Simons theory that lives on a $B$-brane on $\C^k \subset \C^5$. 
\end{proposition*}
Baulieu's result is the special case of this when $k = 5$.  We leverage Baulieu's result to prove the general case. 

\subsection{Residual supersymmetry}
 
Any twist of type IIB supergravity must have some residual supersymmetry consisting of the $Q$-cohomology of the original $(2,0)$ supersymmetry algebra (where $Q$ is the supercharge we use to twist). We find that BCOV theory has this residual supersymmetry.
 
As we have mentioned, we prove that the theory living on a topological $B$-brane on $\C^k$ is the holomorphic twist of the theory living on a physical $D_{2k-1}$-brane.  This theory has some residual supersymmetry given by the $Q$-cohomology of the original supersymmetry algebra of the theory. We calculate this residual supersymmetry algebra and its action on the twisted theory (which is the holomorphic Chern-Simons theory living on the topological $B$-brane). We find that residual supersymmetry is precisely the supersymmetries of BCOV theory which preserve the given brane.

\subsection{Open-closed maps as a check}
The strongest direct match between BCOV theory and twisted type IIB supergravity arises by thinking about open strings ending on $D9$ branes. Baulieu proved the open-string analog of our conjecture: he calculated the $SU(5)$-invariant twist of the $D9$ brane gauge theory, and showed that it is holomorphic Chern-Simons theory, which is the open-string field theory one would find for the space-filling brane in the topological $B$-model on $\C^5$.  
 
Now, in \cite{CosLi15} it was explained that fields of BCOV theory (with their linearized BRST differential) is quasi-isomorphic to the cochain complex of first-order local deformations of the holomorphic Chern-Simons action by single trace operators.  

One believes, on string theory grounds, that type IIB supergravity can be coupled to the $D9$ brane gauge theory by single trace operators. It follows that, whatever twisted type IIB supergravity is, it can be coupled to the twisted $D9$ brane gauge theory, that is, to holomorphic Chern-Simons.  

Since BCOV theory is the universal thing that can be coupled to holomorphic Chern-Simons theory, we find, by this abstract argument,  a map from the fields of twisted IIB supergravity,  to fields of BCOV theory.  This map can be pictured as  
$$
\boxed{\text{twisted type IIB supergravity}} \stackrel{\text{via coupling to HCS}}{\longrightarrow} \boxed{\text{BCOV theory=universal deformation of HCS}}.
$$

In fact, one can extend this argument to the beyond the linearized equations of motion. It was shown in \cite{CosLi15} that a field of BCOV theory satisfies the equations of motion if and only if the corresponding deformation of the holomorphic Chern-Simons action satisfies the classical master equation.  A field of twisted type IIB supergravity which satisfies the equations of motion will lead to such a deformation of holomorphic Chern-Simons,  and so to a field of BCOV theory satisfying the equations of motion.

This description of the relationship between BCOV and twisted supergravity is rather abstract.  To go from this description to a map between (say) BPS solutions of the  supergravity equations of motion and solutions of the equations of motion of BCOV theory is not easy. One would have to analyze the supersymmetric deformation of the $D9$ brane gauge theory that arises from the BPS supergravity solution, then see what this deformation does when twisted to become a deformation of holomorphic Chern-Simons theory, and finally match this deformation of holomorphic Chern-Simons theory to a solution to the BCOV equations of motion.  

Let us also mention some evidence that arises from the AdS/CFT correspondence, and leads to our conjectural description of the twist of type IIB supergravity on $AdS_5 \times S^5$.  We have seen that the theory living on the topological $B$-brane $\C^2 \subset \C^5$ is the holomorphic twist of $\mscr{N}=4$ supersymmetric gauge theory. Thus, the AdS dual of this should be the closed-string topological $B$-model on $\C^5 \setminus \C^2$, deformed in some way by the presence of the branes in the dual picture. It remains to check that the necessary deformation is the one we mentioned above, whereby one introduces a $5$-form flux into the fields of  BCOV theory. We will explain this point in detail elsewhere. 

\section{Residual supersymmetry in BCOV theory} \label{section-SUSY}
\subsection{Residual supersymmetry of twisted supergravity theories}
If BCOV theory is the holomorphic twist of type IIB supergravity, one expects to find some residual supersymmetry present in the fields of BCOV theory. Let us discuss what one expects to find for a general twisted supergravity theory, and then focus on BCOV theory. Suppose we have some BPS background for any supergravity theory. Let $\mf{g}$ denote the super-algebra of bosonic and fermionic symmetries of this supergravity background. Since we are working in Euclidean signature, we need to take $\mf{g}$ to be a complex super Lie algebra, because we may not have a reality condition for the spinors which form the supersymmetries.  Thus, the bosonic part of $\mf{g}$ will contain  the complexified Lie algebra of isometries of our supergravity background (which preserve whatever other bosonic fields take a non-zero value). For example, for flat space background of $\mscr{N}=1$ supergravity in $4$ dimensions, $\mf{g}$ will be the complexification of the $\mscr{N}=1$  super-Poinca\'e algebra\footnote{Or rather a central extension of this} , and for type IIB on $AdS_5 \times S^5$, $\mf{g}$ will be $\mf{psl}(4 \mid 4,\C)$.

Let us assume that our supergravity background has a generalized Killing spinor of square zero.  This will be an element $Q \in \mf{g}$ with $Q^2=\Gamma(Q, Q) = 0$.  Then, we can form a twisted supergravity background by setting the bosonic ghost to take value $Q$.  We expect that the $Q$-cohomology of $\mf{g}$ will be symmetries of the twisted background, and thus appear as ghosts.  We will refer to the $Q$-cohomology of $\mf{g}$ as the \emph{residual} supersymmetry of the twisted background.  In each of the examples of supergravity we study, we will explain how the residual supersymmetry algebra appears as ghosts in our conjectural description of the twisted supergravity theory.

Some care is required here, however.  From abstract principles, all we know is that there will be a Lie algebra homomorphism (or maybe $L_{\infty}$ map) from $\mf{g}$ to the ghosts of the twisted supergravity theory. It could be that the twisted background could have more symmetries, and in practise this often happens. For example, the following can hold. Suppose we have a supergravity background which is a manifold $M$ where the only non-zero bosonic field is the metric $g$, and for which there is a square-zero Killing spinor $Q$.  Suppose that $V$ is a vector field on $M$ such that the Lie derivative $\mc{L}_V g$ is not zero, but is of the form
$$
\mc{L}_V g = Q (\Psi)
$$ 
where $\Psi$ is the gravitino.  Then, although $V$ may not be an isometry, it is modulo $Q$-exact terms and so will become a symmetry of the twisted background.

In a similar way, it is not at all clear that the map from the $Q$-cohomology of $\mf{g}$ to the symmetries of the twisted background will be injective. It could happen that an element of $\mf{g}$ is made $Q$-exact by some vector field which is an isometry modulo $Q$-exact terms.   

Another subtle issue, which we will deal with carefully, relates to extensions of the supersymmetry algebra. It is well known \cite{Tow95} that the $10$-dimensional $(2,0)$ supertranslation algebra admits central extensions corresponding to the $D$-branes, the $NS5$ brane and the fundamental string.   The residual supersymmetry we find will be the $Q$-cohomology not just of the $(2,0)$ algebra, but of its central extension corresponding to the fundamental string.  We argue that this is precisely what one expects from string theory considerations.   

\subsection{Supersymmetry algebras in $10$ dimensions}
Let us first recall the definition of the supersymmetry algebras in $10$ dimensions.  We will always work with complex spin representations, because we work in Euclidean signature. Correspondingly, our Lie algebras of rotations and translations need to be complexified.

The Lie algebra $\mf{so}(10,\C)$ admits two irreducible spin representations $S_+$ and $S_-$. We let $V= \C^{10}$ denote the vector representation.  There are also representations corresponding to exterior powers of $V$. We will often use the notation $\Omega^{i}_{const}(\R^{10})$ to indicate the space of constant complex coefficient $i$-forms on $\R^{10}$, as a representation of $\mf{so}(10,\C)$.  We will also use the notation $\Omega^{i,j}_{const}(\C^5)$ to denote constant-coefficient $(i,j)$ forms, which are a representation of $\mf{sl}(5,\C) \subset \mf{so}(10,\C)$.   

We are also interested in self-dual $5$-forms.  In Euclidean signature, $\ast^2 = -1$ on $\Omega^5(\R^{10})$, so the self-duality condition doesn't make sense. (Here we are working with forms with complex coefficients and using the complex-linear extension of the Hodge star operator). Instead, we can look for forms which are in the $+i$ eigenspace of $\ast$. We will refer to such forms as $\Omega^5_+$, and refer to such forms with constant-coefficients as $\Omega^{5}_{+,const}$. (At some stage, we will probably use the term self-dual form out of habit).  We can write the self-duality condition in complex language by saying that
$$
\Omega^{5}_+ = \Omega^{5,0} \oplus \Omega^{3,2} \oplus \Omega^{1,4}. 
$$

We have the following isomorphisms of representations of $\mf{so}(10,\C)$:
\begin{align*} 
 \Sym^2 S_+ & \iso V \oplus \Omega^{5}_{+,const}\\
 \wedge^2 S_+ & \iso \Omega^{3}_{const} \\
 S_+ \otimes S_- & \iso \Omega^{0}_{const} \oplus \Omega^{2}_{const} \oplus \Omega^{4}_{const}\\ 
\Sym^2 S_- & \iso V \oplus \Omega^{5}_{-,const}. 
\end{align*}
It is, of course, no coincidence that the Ramond-Ramond field strengths of type IIB supergravity appear as the decomposition of $S_+ \otimes S_+$ into irreducible representations, whereas those of type IIA appear in the decomposition of $S_+ \otimes S_-$. 

We will refer to the various maps from tensor powers of $S_+$ and $S_-$ to spaces of forms using notation like
\begin{align*} 
 \Gamma = \Gamma_{\Omega^1} :  S_+ \otimes S_+ & \to V = \Omega^1_{const} \\
\Gamma_{\Omega^5_+} : S_+ \otimes S_+ & \to \Omega^5_{+,const}.  
\end{align*}
In the case that the map lands in the vector representation, we will often just use the notation $\Gamma$ instead of $\Gamma_{\Omega^1}$.  Of course, these maps of $\op{so}(10,\C)$ representations are the representation-theoretic interpretation of the $\Gamma$-matrices commonly used in the physics literature. 

We can decompose $S_+$ and $S_-$ into irreducible representations of $\mf{sl}(5)$ as follows:
\begin{align*} 
 S_+ &\iso \Omega^{0,ev}_{const} \iso \Omega^{odd,0}_{const} \\
S_- & \iso \Omega^{0,odd}_{const} \iso \Omega^{ev,0}_{const}. 
\end{align*}

To describe the $(2,0)$ supersymmetry algebra, fix a copy of $\C^2$ with a non-degenerate symmetric inner product, and fix an orthonormal basis $e_1,e_2$ of $\C^2$. This $\C^2$ is acted on by the group $SO(2,\C)_R$, which is part\footnote{This part of the $R$-symmetry group is the part that is easy to see. The full $R$-symmetry Lie algebra is $\mf{sl}(2,\C)$ when we work in perturbation theory. Non-perturbatively, integrality conditions on the fields break this to $SL(2,\Z)$.  } of the $R$-symmetry group of the theory. Then,
$$
\mc{T}^{2,0} = V \oplus \Pi (S_+ \otimes \C^2).
$$
where the vectors $V$ are central and the Lie bracket on the spinors is obtained by combining the map $\Gamma : \Sym^2 S_+ \to V$ with the inner product on $\C^2$.  $\Pi$ is the parity changing symbol.

In a similar way, the $(1,1)$ supertranslation algebra is
$$
\mc{T}^{1,1} = V \oplus \Pi S_+ \oplus \Pi S_-
$$
where the only non-zero commutators are the maps $\Gamma : \Sym^2 S_+ \to V$ and $\Gamma : \Sym^2 S_- \to V$.  

The super-Poincar\'e algebra is defined as the semi-direct product of the super-translation Lie algebras with $\mf{so}(10,\C)$.

\subsection{Central extensions of supersymmetry algebras}
We will consider two different languages for desecribing the central extensions of supersymmetric algebras. One is quite classical (see \cite{Tow95, deW02} and references therein) whereby the super-translation Lie algebra is centrally extended by vector spaces of constant-coefficient forms.  Another point of view on this has been presented by Baez-Huerta \cite{BaeHue11}, Fiorenza-Sati-Schreiber \cite{FioSatSch15} and Sati-Schreiber-Stasheff \cite{SatSchSta09}, following Castellani, d'Auria and Fr\'e \cite{CasdAuFre91}.  These authors construct one-dimensional central extensions of the super-Poincar\'e algebras, but as $L_\infty$ algebras instead of as ordinary Lie algebras. We will translate between these two points of view.  

Let us first discuss the more classical way of looking at the central extensions. We will focus on the case of type IIB supergravity, with $(2,0)$ supersymmetry.   As we have seen above, the tensor square of $S_+$ decomposes as 
$$
S_+ \otimes S_+ = \Omega^{1}_{const} \oplus \Omega^{3}_{const} \oplus \Omega^{5}_{+,const}. $$
The first and last summand are symmetric, where as the middle summand $\Omega^{3}_{const}$ is anti-symmetric.  

For any symmetric invariant pairing $\omega$ on $\C^2$, we can define a $1$-form valued central extension of the $(2,0)$ supertranslation algebra as follows. The extension will be of the form $\mc{T}^{(2,0)} \oplus \Omega^{1}_{const}$.  The extra term in the commutator, which lands in $\Omega^1_{const}$, will be 
$$
[\psi_1 \otimes e_i, \psi_2 \otimes e_j] = \Gamma_{\Omega^1}(\psi_1 \otimes \psi_2) \omega(e_i,e_j)
$$
where $\psi_i \in S_+$ and $e_i$ are a basis (as before) of $\C^2$.  Of course, this formula mimics the definition of the supertranslation algebra $\mc{T}^{(2,0)}$.  

Similarly, given $\omega$ as above which is symmetric, we can define an $\Omega^{5}_{+,const}$-valued central extension of $\mc{T}^{(2,0)}$ by declaring that the extra term in the commutator, valued in $\Omega^{5}_{+}$, is 
$$
[\psi_1 \otimes e_i, \psi_2 \otimes e_j] = \Gamma_{\Omega^{5}_+} (\psi_1 \otimes \psi_2) \omega(e_i,e_j).
$$ 
Finally, given an anti-symmetric pairing $\eta$ on $\C^2$, we can define an $\Omega^3_{const}$-valued central extension where the extra term in the commutator is
$$
[\psi_1 \otimes e_i, \psi_2 \otimes e_j] = \Gamma_{\Omega^3} (\psi_1 \otimes \psi_2) \eta(e_i,e_j). 
$$

Let us now discuss the meaning of these central extensions in terms of branes and strings. Let us start with the case of the vector-valued central extensions, which correspond to two-dimensional extended objects (strings or $D1$ branes).  The space of symmetric pairings on $\C^2$ is three dimensional, however, in defining the supertranslation algebra $\mc{T}^{(2,0)}$ we have already used a symmetric pairing for which the basis $e_1,e_2$ is orthonormal.  The cocycle giving the central extension corresponding to this invariant pairing is exact, so the central extension is equivalent to a trivial one. There is a two-dimensional space of non-trivial central extensions, corresponding to the two remaining symmetric pairings on $\C^2$. We can think of the two linearly independent central extensions as being related to the $D1$ brane and to the fundamental string.\footnote{It is somewhat arbitrary what is the fundamental string and what is the $D1$ brane, as these are related by the $S$-duality group $SL(2,\Z)$.}.   It is convenient to take the central extension defined by the pairing $\omega(e_i,e_i) = 0$ and $\omega(e_1,e_2) = 1$ and think of this as being related to the fundamental string.

Every element $\psi \otimes e_i$ in the supertranslation Lie algebra will induce a Noether current on the fundamental string mapping to $\R^{10}$.  We will denote this by $\mscr{C}(\psi \otimes e_i)$.   If the string was compatible with the supertranslation algebra $\mc{T}^{(2,0)}$ we would expect that the commutator of these Noether currents would be the current for translation on space-time, that is
$$
[\mscr{C}(\psi_1 \otimes e_i), \mscr{C}(\psi_2 \otimes e_j)] = \mscr{C} ( \Gamma(\psi_1 \otimes \psi_2) \delta_{ij} )
$$ 
where for a vector $v \in V$ we use the notation $\mscr{C}(v)$ to indicate the current on the string coming from translation in the direction $v$. 

However, we find an extra term in the commutator of these currents corresponding to the central extension.  We find that 
$$
[\mscr{C}(\psi_1 \otimes e_i), \mscr{C}(\psi_2 \otimes e_j)] = \mscr{C}(\Gamma(\psi_1 \otimes \psi_2) \delta_{ij} ) + \int \Gamma_{\Omega^1}(\psi_1 \otimes \psi_2) \omega(e_i,e_j). 
$$
The second term on the right hand side involves the integral of the constant-coefficient $1$-form $\Gamma_{\Omega^1}(\psi_1 \otimes \psi_2)$ over a codimension $1$ submanifold in the string worldsheet.  We can take the string worldsheet to be the cylinder $\R \times S^1$ for this discussion, where the Noether currents and this $1$-form are integrated over a circle $t \times S^1$. Because the Noether currents are conserved, and this $1$-form is closed, it doesn't matter what value of $t$ we choose.

In a similar way, the central extensions valued in $3$-forms describe an extra term in the commutator of the Noether currents for the supertranslation algebra on a $D3$ brane.  The extension valued in $\Omega^5_{+,const}$ plays the same role for either the $D5$ brane or the $NS5$ brane (or more generally some $(p,q)$ $5$-brane), depending on the choice of symmetric pairing on $\C^2$.  

Note that we are being very cavalier about integrality conditions: only certain special values of the symmetric pairing will correspond to an extended object. 

\subsection{The $L_\infty$ formulation of the central extensions}
Now let us explain how  to formulate similar central extensions in the language of $L_\infty$ algebras, and explain how the two formulations are related.  

Let us discuss, for concreteness, the central extension corresponding to the fundamental string.  Given an element $\psi_1,\psi_2 \in S_+$, we can pair the one-form $\Gamma_{\Omega^1}(\psi_1 \otimes \psi_2) $ with a vector $v \in V$ to get a tri-linear function on the supertranslation algebra
$$
\mu_1 ( \psi_1 \otimes e_i,\psi_2 \otimes e_j,v) = \Gamma_{\Omega^1}(\psi_1 \otimes \psi_2) (v) \omega(e_i,e_j) 
$$
where as above $\omega$ is a symmetric pairing on $\C^2$.  This trilinear functional is symmetric on the spinors, and so can be viewed as an element of the Chevalley-Eilenberg cochain complex $C^3( \mc{T}^{(2,0)})$.  In the same way, we can use e.g.\ the central extension valued in $\Omega^5_+$ to built an element of $C^7(\mc{T}(2,0))$.  

Sati, Schreiber and Stasheff \cite{SatSchSta09} have verified that these Chevalley-Eilenberg cochains are closed but not exact. One can interpret closed elements of $C^k(\mc{T}^{(2,0)})$ as giving one-dimensional $L_\infty$ central extensions of the Lie algebra $\mc{T}^{(2,0)}$ by a one-dimensional vector space $\C[k-2]$ in degree $2-k$.  The only extra bracket in this $L_\infty$ central extension is the map
$$
l_k  : (\mc{T}^{(2,0)})^{\otimes k} \to \C[k-2]
$$ 
given by the cochain in $C^k(\mc{T}^{(2,0)})$. The statement that our cochain is closed is equivalent to the $L_\infty$ axiom.

In this way, Sati, Schreiber and Stasheff construct $L_\infty$ central extensions of the supertranslation algebra $\mc{T}^{(2,0)}$, from the same data that we used to describe the form-valued central extension. 

Let us explain how these two constructions are related.  To define a current on manifolds of dimension $k+1$ mapping to $\R^{10}$, we need a closed $k$-form on $\R^{10}$.  The central extensions we have discussed are valued in constant coefficient $k$-forms, but we can view them as being valued in closed $k$-forms. Thus, for the $(2,0)$ supersymmetry algebra, we have cocycles
\begin{align*} 
 C_{Dk}& \in C^{2} (\mc{T}^{(2,0)}, \Omega^k_{closed}(\R^{10}) )\\ 
 C_{FS} & \in C^{2} (\mc{T}^{(2,0)}, \Omega^1_{closed}(\R^{10}) )\\
C_{NS5} & \in C^{2} (\mc{T}^{(2,0)}, \Omega^5_{closed}(\R^{10}) ),
\end{align*}
corresponding to $D_k$ branes where $k = 1,3,5$, the fundamental string, and the $NS5$ brane. 

Let us focus again on the central extension corresponding to the fundamental string.  Since our space-time is flat, one could argue that this central extension does not play a role. As, for spinors $\psi_i$, the constant one-form $\Gamma_{\Omega^1}(\psi_1\otimes \psi_2)$ is exact.  It's made exact by a linear functional $\til{\Gamma}(\psi_1 \otimes \psi_2)$. Therefore the integral of Noether current on the string coming from this one-form will be zero. 

However, things are a little more subtle: the linear function $\til{\Gamma}(\psi_1\otimes \psi_2)$ making this $1$-form exact is not translation invariant.  The failure of it to be translation invariant is what leads to the $L_\infty$ central extension.   

Let us discuss explain this more formally. Instead of considering the central extensions of the supertranslation algebra as being valued in forms, we will view them as being valued in the whole de Rham  complex with a shift. For example, we will think of the central extension corresponding to the fundamental string as being valued in $\Omega^\ast(\R^{10})[1]$, where the $[1]$ indicates that we have shifted the de Rham complex so that $1$-forms are in degree $0$.  

This is reasonable, because the integral of the  Noether current on the string associated to the $1$-form valued central extension only depends on the cohomology class of the $1$-form.

 For any Lie algebra, one can consider Chevalley-Eilenberg cochains valued in some cochain complex of representations. The space of $k$-forms on $\R^{10}$ is a representation of $\mc{T}^{(2,0)}$, where the bosonic part (consisting of translations) acts by Lie derivative and the fermionic part acts by zero.   In particular, we can form $C^\ast(\mc{T}^{(2,0)}, \Omega^\ast(\R^{10}))$.  This is the total complex of the double complex
$$
C^\ast(\mc{T}^{(2,0)}, \Omega^0(\R^{10})) \xto{\d_{dR}} C^\ast(\mc{T}^{(2,0)}, \Omega^1(\R^{10}) \xto{\d_{dR}} \dots
$$ 
where the horizontal differential is the de Rham differential and the vertical one is the Chevalley-Eilenberg differential.  The degree in the total  complex is the sum of the de Rham degree and the Chevalley-Eilenberg degree.

There is a natural cochain map
$$
C^2(\mc{T}^{(2,0)}, \Omega^{k}_{closed}) \to C^{2+k}( \mc{T}^{(2,0)}, \Omega^\ast(\R^{10}) )
$$
coming from the map from closed $k$-forms to  the de Rham complex.

In this way, the element of  $C^2(\mc{T}^{(2,0)}, \Omega^1_{closed})$ corresonding to the fundamental string can be viewed as a cocycle in $C^3(\mc{T}^{(2,0)}, \Omega^\ast(\R^{10}))$.

Now, the de Rham complex of $\R^{10}$ is a resolution of the vector space $\C$.  Therefore there is a quasi-isomorphism\footnote{Meaning a cochain map that induces an isomorphism on cohomology} 
$$
C^\ast(\mc{T}^{(2,0)},\C) \to C^\ast(\mc{T}^{(2,0)}, \Omega^\ast(\R^{10}) ). 
$$
In this way, the form-valued cocycles $C_{FS}, C_{Dk}, C_{NS5}$ are cohomologous to elements 
\begin{align*} 
\til{C}_{Dk} &\in  C^{k+2}(\mc{T}^{(2,0)},\C)\\ 
\til{C}_{FS} &\in  C^{3}(\mc{T}^{(2,0)},\C)\\
 \til{C}_{NS5} &\in  C^{7}(\mc{T}^{(2,0)},\C). 
\end{align*} 
These elements are precisely the cocycles considered by Sati-Schreiber-Stasheff, Baez-Huerta, and d'Auria-Fr\'e, defining the $L_\infty$ central extensions.

\subsection{$Q$-cohomology of the centrally extended supersymmetry algebra}
The  above argument suggests that the ``correct'' local supersymmetry algebra of the type IIB superstring is built from the central extension of $\mc{T}^{(2,0)}$ by the de Rham complex corresponding to the fundamental string, together with the Lorentz algebra $\mf{so}(10,\C)$. Following \cite{FioSatSch15} let us denote the semi-direct product of this $\Omega^\ast(\R^{10})[1]$-valued central extension with $\mf{so}(10,\C)$ by $\mf{superstring}_{IIB}$.  Thus, as a graded super-vector space, we have
$$
\mf{superstring}_{IIB} = V \oplus \mf{so}(10,\C) \oplus \Pi ( S_+ \otimes \C^2) \oplus \Omega^\ast(\R^{10})[1]. 
$$ 
Note that we view this as a $\Z \times \Z/2$ graded algebra.  Recall also that $\C^2$ has a  basis $e_1,e_2$, the ordinary commutator of two supercharges involves the inner product $\ip{e_i,e_j} = \delta_{ij}$ on $\C^2$, and the central extension involves the inner product $\omega$ on $\C^2$ defined by $\omega(e_i,e_i) = 0$ and $\omega(e_1,e_2) = 1$.

Fix $\Psi \in S_+$ such that $\Gamma(\Psi,\Psi) = 0 \in V$ and consider the supercharge 
$$Q = \Psi \otimes e_1 \in \mf{superstring}_{IIB}.$$ 
Then $Q^2 = 0$, and we can consider the cohomology of $\mf{superstring}_{IIB}$ with respect to $Q$, which is the residual supersymmetry of our twisted IIB supergravity theory.  We will calculate this cohomology, and find that it appears in BCOV theory. We will denote this cohomology by $\mf{superstring}_{IIB}^{Q}$. 

Let us first introduce some additional notation.  The image $[Q,-]$ inside $V$ is a $5$-complex dimensional space which we call $V^{0,1}$. Further, $V^{0,1}$ is isotropic with respect to the inner product on $V$. Since $V$ is the complexification of a real vector space $V_{\R} = \R^{10}$, we can take the complex conjugate of $V^{0,1}$. Let us denote this complex conjugate by $V^{1,0}$. Then we have
$$
V = V^{1,0} \oplus V^{0,1}. 
$$
In particular, $Q$ induces a complex structure on $V_{\R} = \R^{10}$.  
 
The stabilizer of $Q$ inside $\mf{so}(10,\C)$ is a parabolic Lie algebra $\op{Stab}(Q) \subset \mf{so}(10,\C)$ whose Levi factor is a copy of $\mf{sl}(5,\C) \subset \mf{so}(10,\C)$.  
\begin{lemma}
The $Q$-cohomology $\mf{superstring}_{IIB}^{Q}$ of $\mf{superstring}_{IIB}$ is the super Lie algebra whose underlying graded vector space is
$$
V^{1,0} \oplus \op{Stab}(Q) \oplus \Pi \left( S_+ \otimes e_2\right) \oplus \pi \C \cdot c 
$$
where the element $c$ is central.  The Lie bracket is as follows. First, this Lie algebra is a semi-direct product of $\op{Stab}(Q)$ with $V^{1,0} \oplus \pi \left( S_+ \otimes e_2\right) \oplus \pi \C \cdot c$.  The commutator of two elements
$$
\psi_1,\psi_2 \in S_+ = S_+ \otimes e_2
$$ is
$$
[\psi_1,\psi_2] = \pi_{V^{1,0}} \Gamma (\psi_1 \otimes \psi_2)
$$ 
where $\pi_{V^{1,0}}$ is the projection from $V$ to $V^{1,0}$. 

Finally, the commutator of an element $v \in V^{1,0}$ with $\psi \in S_+$ is
$$
[v,\psi] = c \ip{\Gamma(\Psi \otimes \psi), v}_{V}
$$
where recall that $\Psi \in S_+$ is such that $Q = \Psi \otimes e_1$. 
\end{lemma}
\begin{proof}
Let us first compute the $Q$-cohomology of the non-centrally extended algebra
$$
\mf{siso}_{IIB} = \mf{so}(10,\C) \oplus V \oplus \Pi ( S_+ \otimes \C^2 ) .
$$
The element $Q = \Psi \otimes e_1$ commutes with $S_+ \otimes e_2$, and the image of $[Q,-]$ does not intersect $S_+ \otimes e_2$. Therefore, $S_+ \otimes e_2$ survives to cohomology.  Further, the image of $[Q,-]$ in $V$ is by definition $V^{0,1}$ so that only $V^{1,0} \subset V$ survives in the cohomology.  It is also clear that the kernel of $[Q,-]$ in $\mf{so}(10,\C)$ is $\op{Stab}(Q)$. Thus, to prove the statement (without the central extension) we need to show that $S_+ \otimes e_1$ is killed.  The map 
$$[Q,-] : S_+ \otimes e_1 \to V^{0,1}$$
is surjective, and has $11$ dimensional kernel. It is easy to verify that the kernel of this map is precisely the image of the map
$$
[Q,-] : \mf{so}(10,\C) \to S_+ \otimes e_1.
$$

Thus, we have verified that the cohomology is indeed $\op{Stab}(Q) \oplus V^{1,0} \oplus \Pi (S_+ \otimes e_2) \oplus \pi \C \cdot c$ where $c$ is the central element. We have also verified that the Lie bracket is what we claimed it is, modulo the central term. It remains to calculate the term in the Lie bracket involving $c$. In fact, in the course of doing this, we will verify that there are no higher $L_\infty$ brackets.

Let us represent the $Q$-cohomology by a subspace of the cochain complex $\mf{superstring}_{IIB}$ with differential the internal differential on $\mf{superstring}_{IIB}$ plus $[Q,-]$. Recall that
$$
\mf{superstring}_{IIB} = \mf{so}(10,\C) \oplus V \oplus \Pi (S_+ \otimes \C^2) \oplus \Omega^\ast(\R^{10}) [1] 
$$ 
and the differential is $\d_{dR} + [Q,-]$.   The elements in $\op{Stab}(Q)$ and $V^{1,0}$ can be embedded in $\mf{so}(10,\C)$ and $V$ in the evident way, and become closed elements of this cochain complex.  The central element $c$ is embedded as the function $-1$ in $\Omega^0(\R^{10})$. We need to find cochain representatives for the remaining space $S_+ \otimes e_2$. The obvious embedding of $S_+ \otimes e_2$ in the fermions of $\mf{superstring}_{IIB}$ is not invariant under $[Q,-]$, because the central extension gives us a term in the differential of the form 
\begin{align*} 
[Q,-] :  S_+ \otimes e_2 &\to \Omega^1(\R^{10}) \\
  \psi \otimes e_2 & \mapsto \Gamma_{\Omega^1}(\Psi,\psi). 
\end{align*}
Let us break the translational symmetry of $\R^{10}$, and fix the origin $0 \in \R^{10}$.  We can then define a map
$$ 
 H : \Omega^1(\R^{10})  \to \Omega^0(\R^{10})
$$  
which sends a one-form $\omega$ to the function defined by
$$
\left( H \omega \right)(x) = \int_{0}^{x} \omega
$$
where we integrate over the straight-line path from $0$ to $x$. 

Evidently if $\d_{dR} \omega = 0$ then 
$$
\d_{dR} H \omega = \omega. 
$$
It follows that we can lift a cohomology class $\psi \in S_+ \otimes e_2$ to a closed element at the cochain level by the formula
$$
L(\psi) = \psi \otimes e_2 - H \Gamma_{\Omega^1} (\Psi \otimes \psi) \in S_+ \otimes e_2 \oplus \Omega^0[1]. 
$$
Finally, we can calculate the commutators betweeen our cochain representatives of the cohomology classes.  We find that our cochain representatives form a sub dg Lie algebra, and that modulo the central term the commutators are the obvious ones. However, since $L(\psi)$ is not translation-invariant, we find an extra commutator between a vector $v \in V^{1,0}$ and an element $\psi \in S_+ \otimes e_2$ of the form 
$$
[v, L(\psi)] = - \ip{v,\Gamma(\Psi \otimes \psi)}= c \ip{v,\Gamma(\Psi \otimes \psi)}  
$$ 
as desired. 
\end{proof}

\subsection{Embedding $\mf{superstring}_{IIB}^Q$ into BCOV theory} As a first check of our conjectured relation between twisted IIB supergravity and BCOV theory, we will show that the $Q$-cohomology $\mf{superstring}_{IIB}^{Q}$ of $\mf{superstring}_{IIB}$ appears naturally by fields of BCOV theory.  We also need to match the commutators in the supersymmetry algebra with corresponding commutators in BCOV theory. Before we do this, we need to explain what the commutators in BCOV theory are.  

As above, let $\PV(\C^5)\left\llbracket t\right\rrbracket[2]$ denote the fields of BCOV theory on $\C^5$.  This has a differential $\dbar + t \partial$.  The BCOV interaction induces an $L_\infty$ structure on a shift by one of these fields, namely $\PV(\C^5)\left\llbracket t\right\rrbracket[1]$. The $L_\infty$ structure is such that the Maurer-Cartan equation for this $L_\infty$ structure is the equations of motion of BCOV theory.

We will write down a map of Lie algebras from the $Q$-cohomology of the supersymmetry algebra to the cohomology (with respect to $\dbar + t \partial$) of $\PV(\C^5)\left\llbracket t\right\rrbracket[1]$, with its induced Lie algebra structure. On $\dbar + t \partial$ cohomology, the Lie algebra structure coincides with the one coming from the Schouten bracket of poly-vector fields.  

To do this, we will first decompose the $Q$-cohomology of $\mf{superstring}_{IIB}$ in terms of representations of $\mf{sl}(5,\C) \subset \op{Stab}(Q)$.  We will write the decomposition in terms of constant-coefficent forms and polyvector fields forms on $\R^{10} = \C^5$, using the complex structure determined  by $Q$.  We find that 
\begin{align*} 
 \op{Stab}(Q) &\iso \mf{sl}(5,\C) \oplus \PV^{3,0}_{const} \iso \mf{sl}(5,\C) \oplus \Omega^{2,0}_{const} . \\
S_+  &\iso \Omega^{odd,0}_{const} \iso \Omega^{1,0}_{const} \oplus \PV^{2,0}_{const} \oplus \Omega^{0,0}_{const}.  
\end{align*}
Here $\PV_{const}^{i,j}$ denotes poly-vector fields with constant coefficients. Note also that $V^{1,0}=\PV^{1,0}_{const}$. To summarize, we have 
$$
\mf{superstring}_{IIB}^{Q} = \mf{sl}(5,\C) \oplus \PV^{1,0}_{const} \oplus  \PV^{3,0}_{const} \oplus \Pi\left( \Omega^{1,0}_{const} \oplus \PV^{2,0}_{const} \oplus  \Omega^{0,0}_{const} \oplus \C \cdot c\right). 
$$
Note that this Lie algebra actually has \emph{two} fermionic central elements, namely the central element $c$ coming from the central extension before we took $Q$-cohomology, and the unique up to scale $\mf{sl}(5,\C)$ invariant element of $S_+$ (which we have written as $1 \in \Omega^{0,0}_{const}$ in the decomposition above).  The second fermionic central element comes from the element $\Psi \otimes e_2$ in the original supersymmetry algebra, where the $Q = \Psi \otimes e_1$.  

The map from $\mf{superstring}^{Q}_{IIB}$ to the fields of BCOV theory is as follows. 
\begin{align*} 
 A_{ij} \in \mf{sl}(5,\C) & \mapsto \sum A_{ij} z_i \dpa{z_j} \in \PV^{1,0} \subset \PV\left\llbracket t\right\rrbracket \\
\dpa{z_i} \dpa{z_j} \dpa{z_k} \in \PV^{3,0} & \mapsto   \dpa{z_i} \dpa{z_j} \dpa{z_k} \in \PV^{3,0} \subset \PV\left\llbracket t\right\rrbracket \\
\dpa{z_i} \in \PV^{1,0}_{const} = V^{1,0} & \mapsto \dpa{z_i} \in \PV\left\llbracket t\right\rrbracket \\
\d z_i \in \Omega^{1,0}_{const} &\mapsto z_i \in \PV^{0,0} \subset \PV\left\llbracket t\right\rrbracket\\
\dpa{z_i} \dpa{z_j} \in \PV^{2,0}_{const} & \mapsto \dpa{z_i} \dpa{z_j} \in \PV^{2,0} \subset \PV\left\llbracket t\right\rrbracket \\
1 \in \Omega^{0,0}_{const} & \mapsto 0 \\
c & \mapsto 1 \in \PV^0 \subset \PV\left\llbracket t\right\rrbracket.  
\end{align*}
\begin{lemma}
This map defines a Lie algebra homomorphism from $\mf{superstring}_{IIB}^{Q}$ to the space 
$$
\op{Ker} \partial  \subset \PV_{hol}(\C^5)[1]
$$
equipped with its Schouten bracket (where the subscript $hol$ indicates holomorphic polyvector fields.  
\end{lemma}
\begin{proof}
This is a simple calculation. For instance, the commutator
$$
 \{z_i, \partial_{z_j} \partial_{z_k} \} = \delta_{ij} \partial_{z_k} - \delta_{ik} \partial_{z_j} 
$$
shows how to spinors in $S_+ \subset \mf{superstring}_{IIB}^{Q}$ commute to give a vector.  The commutator
$$
\{\dpa{z_i}, z_j\} = \delta_{ij}
$$
corresponds to the central extension in $\mf{superstring}_{IIB}^{Q}$, whereby a vector and a spinor commute to give a multiple of the central element. 
\end{proof} 

 Supersymmetries in $S_+ \subset \mf{superstring}_{IIB}^Q$ become either linear superpotentials $z_i\in \PV^{0,0}$ or else bivectors $\dpa{z_i} \dpa{z_j}\in \PV^{2,0}$.  This tells us that further twists of type IIB supergravity are obtained by considering the topological $B$-model on $\C^5$ either with a linear superpotential, or where some directions are non-commutative, or a mixture of the two.

The linear superpotential is slightly mysterious, because naively one expects that the topological $B$-model with a linear superpotential is trivial, because it has no critical points.  It seems that the topological $B$-model with a linear superpotential is somewhat similar in nature to the Donaldson-Witten twist of $N=2$ supersymmetric gauge theory. That is, it is trivial in perturbation theory, but non-perturbatively it counts certain ``gravitational instantons''. 

The further twists which correspond to introducing a Poisson bivector are more understandable. Recall that the topological $B$-model on a non-commutative space is closely related to the topological $A$-model.  As, the $N=(2,2)$ model with a hyper-K\"ahler target has a $\mbb{P}^1$ of topological twists \cite{VafWit94, KapWit06} with the feature that at $\infty$ in the $\mbb{P}^1$, it is the topological $B$-model; at $0$, it is the $A$-model; at intermediate points, it can be interpreted either as the $A$-model with a non-zero $B$-field or as the $B$-model on a non-commutative space.  The $B$-field in the $A$-model corresponds to the holomorphic symplectic form for a chosen complex structure on the target, which is the inverse to the Poisson tensor defining the non-commutative deformation of the $B$-model. 

Applied to our situation, this discussion shows that (assuming our conjectures) a further twist of type IIB, where two of the $5$ complex directions are made non-commutative, can be described by a mixture of the topological $A$-model (in these two directions) and the topological $B$-model (in the remaining $3$). 

\section{$D$-branes in type IIB and their supersymmetry}\label{section-HH}
In this section we will give a conjectural description of $D$-branes in type IIB and check that they have the right supersymmetry. We do not currently have a good understanding of the $NS5$ brane, and defer discussion of this to subsequent publications.  

Our conjecture, of course, is that the $D_{2k-1}$-brane (for $k = 0,\dots,5$) corresponds to the topological $B$-brane living on the structure sheaf of a copy of $\C^k$ inside of $\C^5$.  Our task is to verify that the theory living on the brane in the topological $B$-model is a holomorphic twist of the maximally supersymmetric gauge theory living on the brane in type IIB. We also need to check that the residual supersymmetries in the twisted supergravity theory that preserve the brane match the residual supersymmetries that one expects in the holomorphic twist of the maximally supersymmetric gauge theory.

\subsection{The theory living on a $B$-brane}
The first thing we need to do is to analyze the theory living on a $D_{2k-1}$ brane on $\C^k \subset \C^5$.   In general, for any Calabi-Yau $X$, a $B$-brane on $X$ is given by a coherent sheaf $E$.  The open-string field theory constructed from the sheaf $E$ is a version of holomorphic Chern-Simons theory.  The space of fields is the open-string state space, whose cohomology is $\op{Ext}(E,E)[1]$. 

Let us apply this to the structure sheaf $\Oo_{\C^k}$ of $\C^k \subset \C^5$. In this case we can view the $\op{Ext}$-algebra as a holomorphic bundle of algebras on $\C^k$.    Standard homological algebra tells us that this bundle is the exterior algebra of the normal bundle of $\C^k$ inside $\C^5$. Thus, the fields of the holomorphic Chern-Simons theory constructed from the $B$-brane on $\C^k$ are $\Omega^{0,\ast}(\C^k)[\eps_1,\dots,\eps_{5-k}][1]$, where the $\eps_i$ are odd variables (which we can take to be of cohomological degree $1$).

   More generally, if we take $N$ copies of the structure sheaf of $\C^k$ inside $\C^5$ we find 
$$ \op{Ext}_{\Oo(\C^5)} (\Oo^N_{\C^k}, \Oo^N_{\C^k}) \simeq \Omega^{0,\ast}(\C^k)[\eps_1,\dots,\eps_{5-k}] \otimes \gl_N [1].$$ 
In the case $k = 3$, this is precisely the field content (in the BV formalism) of the holomorphic twist of $\mscr{N}=4$ Yang-Mills as described in \cite{Cos11b}.  

We can interpret $\Omega^{0,\ast}(\C^k)[\eps_i]$ as being the Dolbeault complex of the complex supermanifold $\C^{k \mid 5-k}$, where we only take the Dolbeault resolution in the even $\C^k$ direction and consider holomorphic functions in the odd directions.  

The action functional is the holomorphic Chern-Simons type action
$$
S(A) = \int_{\C^{k \mid 5-k}}\left( \tfrac{1}{2} \op{Tr} (A \dbar A) + \tfrac{1}{3} \op{Tr} (A^3)\right) \prod \d z_i \prod \d \eps_i
$$
where $A \in \Omega^{0,\ast}(\C^{k \mid 5-k}) \otimes \gl_N[1]$ is the field.  Here the integral is over the supermanifold $\C^{k \mid 5-k}$. Concretely, performing this integral means picking up the term in the field which lives in $\eps_1 \dots \eps_{5-k} \Omega^{0,k}(\C^k)$ and then integrating it in the usual way over $\C^k$.

It is important to note that the $\eps_i$ live in the normal bundle to $\C^k$, and so transform under $\mf{sl}(5-k)$ in the representation which is dual to that given by the functions $z_1,\dots,z_{5-k}$. 

Now we can prove the following (which is an easy consequence of Baulieu's result \cite{Bau10} that we already mentioned).
\begin{lemma}
Holomorphic Chern-Simons theory on the supermanifold $\C^{k \mid 5-k}$ is the holomorphic twist of the maximally supersymmetric gauge theory in dimension $2k$.
\end{lemma}
\begin{proof}
In the case that $k = 5$ this was proved by Baulieu. Precisely,  Baulieu showed in  \cite{Bau10} that the gauge fixed action of the holomorphic twist of $N=1, D=10$ super Yang-Mills theory in the BV formalism becomes the holomorphic Chern-Simons action modulo $Q$-exact terms. The maximally supersymmetric theory in dimension $2k$ is a reduction of $N=1$ gauge theory in dimension $10$, and the process of dimensional reduction commutes with taking the holomorphic twist.  To check the result, we simply need to analyze the dimensional reduction of holomorphic Chern-Simons theory on $\C^5$ to $\C^k$.  The fields of holomorphic Chern-Simons on $\C^5$ are $\Omega^{0,\ast}(\C^5)$. Let us choose coordinates $w_1,\dots,w_k$ and $z_{1}, \dots,z_{5-k}$ on $\C^5$, where our brane lives on the space with coordinates $w_i$.  If we take the fields of holomorphic Chern-Simons on $\C^5$ to  be constant in the $z_i$ directions, we are left with
$$
\Omega^{0,\ast}(\C^k)[\d \zbar_1 \dots \d \zbar_{5-k}].
$$ 
We can identify the $\d \zbar_i$ with the $\eps_i$ in the description of the theory on a $B$-brane given above.  The dimensionally reduced action is the holomorphic Chern-Simons action on $\C^{k \mid 5-k}$ we discussed above. 
\end{proof}

\subsection{Supersymmetries of type IIB and holomorphic twists of $D$-brane gauge theories} 
Next, we will describe how the residual supersymmetries of type IIB supergravity (which we have represented by particular polyvector fields on $\C^5$) act on the $D$-brane gauge theory. This action will be a special case of the coupling of a general polyvector field on $\C^5$  to holomorphic Chern-Simons on $\C^{k \mid 5-k}$.  The explicit formula for the coupling of an arbitrary field of BCOV theory is a little complicated: the corresponding coupling for the space-filling $B$-brane $\C^5$ is presented in detail in \cite{CosLi15}.  Instead of presenting this formula in detail, we will explain the abstract reasons for the existence of such a coupling and then analyze how the polyvector fields corresponding to supersymmetries coupling.  

General categorical results tell us that, for any complex manifold $X$, elements of the Hochschild cohomology of the structure sheaf give rise to (curved) $A_\infty$ deformations of an appropriate dg model for the category of coherent sheaves on $X$.  In particular, we have a map
$$
HH^\ast(\Oo_X) \to HH^\ast (\op{RHom}(E,E)) 
$$ 
for any coherent sheaf $E$ on $X$. (By $HH^\ast(\Oo_X)$ we mean the ``local'' version of Hochschild cohomology, where the Hochschild cochains are supported on the diagonal). Since, by the Hochschild-Kostant-Rosenberg theorem, $HH^\ast(\Oo_X)$ is the same as the algebra $\PV(X)$ of polyvector fields on $X$, we get a map 
$$
\PV(X) \to HH^\ast(\op{RHom}(E,E)).
$$  
In the case that $X$ is Calabi-Yau, this map takes the operator $\partial$ on $\PV(X)$ to the Connes $B$-operator on $HH^\ast(\op{RHom}(E,E))$. We can thus pass to the cyclic complex and get a cochain map
$$
\PV(X)\left\llbracket t\right\rrbracket \to HC^\ast(\op{RHom}(E,E))
$$
where $\PV(X)\left\llbracket t\right\rrbracket$ is equipped with the differential $\dbar + t\partial$ and $HC^\ast$ indicates the cyclic cohomology.  The cyclic cohomology describes deformations of $\op{RHom}(E,E)$ as a curved $A_\infty$ algebra with a trace, and so deformations of the holomorphic Chern-Simons theory constructed from $\op{RHom}(E,E)$. 

Applying this to the case when $E$ is the structure sheaf $\Oo_{\C^k}$ for $\C^k \subset \C^5$,  we find a map 
$$
\PV(\C^5)\left\llbracket t\right\rrbracket \to HC^\ast( \Omega^{0,\ast}(\C^k)[\eps_i] ) 
$$
where on the right hand side the cyclic cohomology describes deformations of $\Omega^{0,\ast}(\C^k)[\eps_i]$ as a dg algebra with a trace.

A further application of the Hochschild-Kostant-Rosenberg theorem tells us that we can identify the cyclic cohomology of $\Omega^{0,\ast}(\C^k)[\eps_i]$ with the space 
$$
\PV^{\ast,\ast}(\C^{k \mid 5-k} )\left\llbracket t\right\rrbracket = \PV^{\ast,\ast}(\C^k)\left \llbracket \eps_i, \dpa{\eps_i}\right \rrbracket \left\llbracket t\right\rrbracket 
$$
where the variables $\dpa{\eps_i}$ are even (and of cohomological degree $0$ if we give the variables $\eps_i$ degree $1$).  This complex, as usual, has differential $\dbar + t \partial$ where $\partial$ is the divergence operator.  

Now, up to inverting the HKR isomorphism, we just need to write down an explicit map from polyvector fields on $\C^5$ to those on $\C^{k \mid 5-k}$.  It turns out that the desired map is easier to write down if we focus on the space $\PV_{hol}(\C^5) =\PV_{hol}^{\ast,0}(\C^5)$ of polyvector fields which are in $0$ Dolbeault degree and are in the kernel of $\dbar$.      Let us use coordinates $w_1,\dots,w_k$ and $z_1,\dots,z_{5-k}$ on $\C^5$ where the brane we consider lives on the space with coordinates $w_i$.  We will use Roman indices for the $w_i$ coordinates and Greek indices for the $z_\alpha$ coordinates.  
\begin{lemma}
The map 
$$
\PV_{hol}(\C^5) \to \PV_{hol}(\C^{k \mid 5-k}) = \PV_{hol}(\C^k) \left \llbracket \eps_\alpha, \dpa{\eps_\alpha} \right \rrbracket
$$
which corresponds to the map
$$
HH^\ast(\Omega^{0,\ast}(\C^5)) \to HH^\ast( \Omega^{0,\ast}(\C^{k \mid 5-k}))
$$
on Hochschild cohomology is the unique continuous map of algebras which sends
\begin{align*} 
  w_i & \mapsto w_i \\
\dpa{w_i} & \mapsto \dpa{w_i} \\
z_\alpha & \mapsto \dpa{\eps_\alpha} \\
\dpa{z_\alpha} & \mapsto \eps_\alpha.
\end{align*}
\end{lemma}
In this statement, we are thinking of holomorphic polyvector fields on $\C^5$ as a completion of the algebra $\C[z_\alpha, w_i, \dpa{z_\alpha}, \dpa{w_i}]$ of holomorphic polyvector fields with polynomial coefficients. 
\begin{proof}

We will first prove the result for the case $k = 0$, and then explain the easy generalization. The Ext algebra of the origin inside $\C^5$ is of course $\C[\eps_1,\dots,\eps_5]$.  We need to explain how a polyvector field on $\C^5$ leads to a polyvector field on $\C^{0 \mid 5}$. 

We will let $V$ denote the vector space $\C^5$.  Then, the Ext algebra of the structure sheaf of the origin is $\wedge^\ast V$, whereas the algebra of polynomial functions on $V$ is $\Sym^\ast V^\vee$.  These algebras are Koszul dual.

It is standard that (once one takes care with completions) Koszul dual algebras have the same Hochschild cohomology.   The Hochschild-Kostant-Rosenberg theorem tells us that the Hochschild cohomology of $\wedge^\ast V$ is the algebra of poly-vector fields on the graded manifold $ V^\vee[-1]$ which is 
$$
\PV( V^\vee[-1]) =  \wedge^\ast V \otimes \what{\Sym}^\ast V^\vee
$$
where $V$ is in degree $1$ and $V^\ast$ is in degree $0$. 

Similarly, the Hochschild cohomology of $\Sym^\ast V^\vee$ is polyvector fields on $V$, which is 
$$
\PV(V) = \Sym^\ast V^\vee \otimes \wedge^\ast V.
$$
A standard (and easy) result in the theory of Koszul duality tells us that the desired map from polyvector fields on $V$ to the Hochschild cohomology of $\wedge^\ast V^\ast$ is the evident map
$$
\Sym^\ast V^\vee \otimes \wedge^\ast V \to \wedge^\ast V \otimes \what{\Sym}^\ast V^\vee.
$$
If, instead of taking the algebra of polynomial polyvector fields on $V$ we took the algebra of holomorphic vector fields, we would get a similar map 
$$
\PV_{hol}(V) = \Oo(V) \otimes \wedge^\ast V \to \what{\Sym}^\ast V^\vee \otimes \wedge^\ast V
$$
obtained by taking the power-series expansion of a holomorphic function in $\Oo(V)$.

A similar argument applies when $k > 0$. 

\end{proof}
Let us now use this lemma to analyze how the supersymmetries present in BCOV theory act on the gauge theory. Let us continue to use coordinates on $\C^5$ given by $w_i,z_\alpha$ where the brane lives on the locus $z_\alpha = 0$.  In BCOV theory, the supersymmetries  are represented by the polyvector  fields $z_\alpha$, $w_i$, $\dpa{z_\alpha} \dpa{z_\beta}$, $\dpa{w_i} \dpa{w_j}$, and $\dpa{z_\alpha} \dpa{w_j}$.  Using the map described above, these become the polyvector fields on $\C^{k \mid 5-k}$ given by $\dpa{\eps_\alpha}$, $w_i$, $\eps_\alpha \eps_\beta$, $\dpa{w_i} \dpa{w_j}$, and $\eps_\alpha \dpa{w_j}$.  We will view each such supersymmetry as giving a first-order deformation of the holomorphic Chern-Simons theory on $\C^{k \mid 5-k}$, and we will analyze each such deformation in turn.  Since the deformations of holomorphic Chern-Simons we find all arise from deformations of the dg algebra $\Omega^{0,\ast}(\C^k)[\eps_\alpha]$ into a curved $A_\infty$ algebra, we will also explain this interpretation.  

The reader should bear in mind that not all supersymmetries of type IIB supergravity correspond to supersymmetries of the $D$-brane gauge theory.  So only some of the operators that we will describe in our twisted $D$-brane theories will arise from supersymmetries of the physical theory.  We will see shortly which ones come from the physical gauge theory and which do not.  
\begin{enumerate} 
 \item The operator $\dpa{\eps_\alpha}$ is an odd derivation of the algebra $\Omega^{0,\ast}(\C^k) [\eps_\alpha]$, deforms it to a new dg algebra by adding $\dpa{\eps_\alpha}$ to the differential. This deforms the holomorphic Chern-Simons action by a quadratic term of the form
$$
\int_{\C^{k \mid 5-k}} \prod \d w_i \prod \d \eps_\alpha \op{Tr} \left( A \dpa{\eps_\alpha} A\right)
$$
where $A \in \Omega^{0,\ast}(\C^k)[\eps_\alpha][1] \otimes\gl_N$ is a field.
\item The operators $w_i$ and $\eps_\alpha \eps_\beta$ both turn  $\Omega^{0,\ast}(\C^k)[\eps_\alpha]$ into a curved $A_\infty$ algebra, where the curving is given by the elements $w_i$ and $\eps_\alpha \eps_\beta$ respectively.  These give rise to linear action functionals deforming the holomorphic CS action of the form
\begin{align*} 
A  &\mapsto \int_{\C^{k \mid 5-k}} \prod \d w_i \prod \d \eps_\alpha \op{Tr}( A) w_i\\ 
 A  &\mapsto \int_{\C^{k \mid 5-k}} \prod \d w_i \prod \d \eps_\alpha \op{Tr}(A) \eps_\alpha \eps_\beta. 
\end{align*}
\item The supersymmetry $\dpa{w_i} \dpa{w_j}$ makes the space $\C^k$ non-commutative in the $i,j$ plane, to first order.  In turn, this deforms the algebra $\Omega^{0,\ast}(\C^k)[\eps_\alpha]$ to a non-commutative algebra (to first order), where the Poisson bracket describing the deformation is simply
$$
\{f,g\} = \dpa{w_i} f \dpa{w_j} g
$$
for $f,g \in \Omega^{0,\ast}(\C^k)[\eps_\alpha]$.

The corresponding action functional is cubic and is of the form
$$
A  \mapsto \int_{\C^{k \mid 5-k}} \prod \d w_i \prod \d \eps_\alpha \op{Tr} \left( A \dpa{w_i}A \dpa{w_j} A \right).  
$$ 
\item Finally the supersymmetry $\eps_\alpha \dpa{w_j}$ is a derivation of $\Omega^{0,\ast}(\C^k)[\eps_\alpha]$, and so deforms the dg algebra to first order by adding $\eps_\alpha \dpa{w_j}$ to the differential.  The corresponding first-order deformation of the action functional is given by
$$
\int_{\C^{k \mid 5-k}} \prod \d w_i \prod \d \eps_\alpha \op{Tr} \left( A \eps_\alpha \dpa{w_j} A\right).
$$
 \end{enumerate}

In physical type IIB string theory, a $D$-brane is preserved by a $\tfrac{1}{2}$-BPS subalgebra of the full $(2,0)$ supersymmetry algebra.  Thus, we would not expect every supersymmetry that we can see in type IIB to correspond to a supersymmetry on the twisted $D$-brane gauge theory. One can ask, which of the supercharges above can correspond to residual supersymmetries of the $D$-brane gauge theory?  It turns out that not all of them can.  For instance, the deformation of holomorphic Chern-Simons on $\C^{k \mid 5-k}$ which makes some of the bosonic directions non-commutative only works if our gauge Lie algebra is $\gl(N)$.  This means that the supersymmetry in BCOV theory on $\C^5$ which implements this deformation can not possibly arise from a supersymmetry of the maximally supersymmetric gauge theory on $\R^{2k}$, because all such supersymmetries can be defined for arbitrary gauge groups. 

In a similar way, we can exclude the supersymmetries which correspond to deforming $\Omega^{0,\ast}(\C^k)[\eps_\alpha]$ into a curved $A_\infty$ algebra (these correspond to the polyvector fields $\eps_\alpha \eps_\beta$ and $w_i$).   In this case the equations of motion -- which are the Maurer-Cartan equation in the curved $A_\infty$ algebra -- read
\begin{align*} 
 \dbar A + \tfrac{1}{2} [A,A] + w_i \op{Id} &= 0\\
 \dbar A + \tfrac{1}{2} [A,A] + \eps_\alpha \eps_\beta \op{Id} &= 0 
\end{align*}
where as above $A \in \Omega^{0,\ast}(\C^k) [\eps_\alpha] \otimes \gl_N[1]$ is a field.  In both cases, the field $A = 0$ does not solve the equations of motion, and if we use $\gl_1$, there are no solutions to the equations of motion. This means that these supersymmetries can not possibly arise from supersymmetries in the physical theory, because in the physical theory the trivial field configuration is always BPS.  

The supersymmetries we have not excluded are the ones of the form $\dpa{\eps_\alpha}$ and $\eps_\alpha \dpa{w_j}$. Note that these commute to translations $\dpa{w_j}$.  We will show that these supersymmetries are precisely the residual supersymmetries present after the holomorphic twist of the maximally supersymmetric gauge theory. More precisely, we will prove the following theorem in Appendix \ref{Appendix-A}. 
\begin{theorem}\label{thm-IIB-SUSY}
The residual supersymmetry algebra which acts on the holomorphic twist of the maximally supersymmetric gauge theory in dimension $2k$ has $(5-k)(k+1)$ elements.  The holomorphic twist is holomorphic Chern-Simons theory on $\C^{k \mid 5-k}$, and the residual supersymmetries act by the vector fields $\eps_\alpha \dpa{w_i}$ and $\dpa{\eps_\alpha}$. 
\end{theorem}
Implicit in the statement is that these vector fields commute with translations in the same way that the residual supersymmetries do, and that they transform in the correct way under the residual $R$-symmetry group $SU(5-k) \subset \op{Spin}(10-2k)$ and under the space-time rotations $SU(k) \subset \op{Spin}(2k)$.

\subsection{Twists of $\mscr{N}=4$ supersymmetric gauge theory}
As an example of this story, we will explain how to realize the $\mbb{P}^1$ of twists of $\mscr{N}=4$ supersymmetric Yang-Mills in terms of supergravity.

As we have seen, the holomorphic twist of the $D3$ brane  theory is holomorphic Chern-Simons theory on $\C^{2 \mid 3}$. If, as before, $w_i$ are bosonic coordinates and $\eps_\alpha$ are fermionic coordinates on this supermanifold, then the residual supersymmetry of $\mscr{N}=4$ Yang-Mills acts by the vector fields $\eps_\alpha \dpa{w_i}$ and $\dpa{\eps_\alpha}$.

The $\mbb{P}^1$ of twists described by Kapustin and Witten have the following interpretation in terms of holomorphic Chern-Simons on $\C^{2 \mid 3}$.  This interpretation was derived in \cite{Cos11b}.  The $\mbb{P}^1$ of twists introduced by Kapustin and Witten is implemented by the family of supercharges
$$
s \left(\eps_1 \dpa{w_1} + \eps_2 \dpa{w_2}  \right) +  t \dpa{\eps_3} 
$$
for complex numbers $s,t$.  There is a $\C^\times$ in the $R$-symmetry group $\mf{sl}(3)$ which scales $\eps_1,\eps_2$ with weight $1$ and $\eps_3$ with weight $-2$.  This $\C^\times$ acts on this family of supercharges by $(s,t) \mapsto (\lambda s, \lambda^2 t)$. It follows that we have a weighted $\mbb{P}^1$ of twists.

\begin{lemma}
This $\mbb{P}^1$ of twists is implemented, in BCOV theory on $\C^5$, by the family of local supersymmetries
$$
s \left( \dpa{z_1} \dpa{w_1} + \dpa{z_2} \dpa{w_2} \right) + t z_3, 
$$
where $z_i$ are coordinates on the normal $\C^3$ to $\C^2 \subset \C^5$. 
\end{lemma}  
\begin{proof}
This follows immediately from our discussion above.  
\end{proof}

\subsection{Quantization}
Now we can quote one of the main results of \cite{CosLi15}.
\begin{theorem*}
BCOV theory on $\C^5$ admits a unique quantization which extends to a quantization of the theory coupling BCOV theory with holomorphic Chern-Simons on $\C^5$, with gauge group $\gl(N \mid N)$, compatible with the embedding $\gl(N \mid N)\into \gl(N+1 \mid N+1)$. 
\end{theorem*} 
This theorem tells us that our candidate for twisted type IIB supergravity can be quantized in perturbation theory, in a canonical way.

Combining the result of Baulieu with that of \cite{DijHeiJefVaf16}, we find that holomorphic Chern-Simons on $\C^5$ with gauge Lie algebra $\mf{gl}(N\mid N)$ is a twist of the theory living on a system of $N$ $D9$ branes and $N$ anti-$D9$-branes.  Thus, we can interpret this theorem as saying that there is a unique quantization of twisted type IIB supergravity which is compatible with  coupling to such a brane-antibrane system.

\section{Charges for $D$-branes}
In the physical string theory, $D$-branes are eletrically and magnetically charged under the Ramond-Ramond fields.  In this section we will see how branes in the $B$-model on $\C^5$ are magnetically charged under certain fields of  BCOV theory. This will allow us to match certain fields of BCOV theory with components of the RR field-strengths of type IIB.  

 We will argue that we can't see how branes are electrically charged in the twisted theory because BCOV theory only knows about the field-strengths of the Ramond-Ramond fields. 

First, let us recall how branes in the physical string theory are magnetically charged under RR fields.  Suppose we have a $D_{2k-1}$-brane on $\R^{2k} \subset \R^{10}$.  We can construct a de Rham current
$$
\d^{-1} \delta_{\R^{2k}} \in \Omega^{10-2k-1}(\R^{10}),
$$ 
where by $\d^{-1}$ we mean $\d^\ast \tr^{-1}$ where $\tr$ is the Laplacian.  Then, the equations of motion for type IIB in the presence of a $D_{2k-1}$ brane have the feature that
$$
F_{10-2k-1} = \d^{-1} \delta_{\R^{2k}}. 
$$
Here by $F_{l}$ we mean the $l$-form which is the field-strength of the RR $l-1$-form which we denote by $A_{l-1}$.  The RR field strengths are not independent, we have
$$
F_{l} = \ast F_{10-l}.
$$
This tells us that the RR forms themselves are not independent; rather, $A_{l-1}$ is the electro-magnetic dual of $A_{10-l-1}$.  In particular, the RR $4$-form is self-dual in this sense, leading to the constraint that $F_{5}$ is a self-dual form. 

One arises at the expression for $F_{10-2k-1}$ in the presence of a $D_{2k-1}$-brane by observing that, in the presence of a $D_{2k-1}$-brane, there are terms in the Lagrangian of the form
$$
\int_{\R^{2k}} A_{2k} + \int \d A_{2k} \ast \d A_{2k} = \int_{\R^{2k}} A_{2k} - \int A_{2k} \d F_{10-2k-1}.  
$$
Varying with respect to $A_{2k}$ tells us that $\d F_{10-2k-1}$ is a $\delta$-function on $\R^{k}$. 

\subsection{}
We will find a similar pattern in BCOV theory.  Consider a brane on $\C^k \subset \C^5$.  We would like to view the brane as deforming the action functional of BCOV theory.  The natural deformation of the action functional will arise by considering the $B$-model topological string on worldsheets which have one boundary on the brane. The first such worldsheet that one encounters is a disc with boundary on the brane, and with a single interior marked point labelled by a field of BCOV theory.

There is a subtlety that arises when one tries to make sense of this $B$-model amplitude. Because we are considering the topological string, we should integrate over the moduli of the worldsheet. In this case, the worldsheet has no moduli, but it has an $S^1$ symmetry, which needs to be taken account of carefully when constructing the amplitude.  

Before we understand what happens when we take the quotient by the $S^1$ action, let us try to understand how things work without taking the quotient.  We can do this by rigidifying the disc by putting a single marked point on the boundary, which we label by the identity operator on the brane.

Then, we find that the closed-open string map is the map
\begin{align*} 
\PV^{\ast,\ast}_c(\C^5) & \mapsto \C\\
\alpha^{i,j} & \mapsto 0 \text{ if } (i,j) \neq (5-k,k)\\
\alpha^{5-k,k} & \mapsto \int_{\C^k} \alpha^{5-k,k} \vee \Omega 
\end{align*}
where $\alpha^{5-k,k} \vee \Omega$ is the $(k,k)$ form obtained by contracting the polyvector field $\alpha^{5-k,k}$ with $\Omega$.

What happens when we take account of the circle rotation?  In a topological field theory, the $S^1$ action on the space of states associated to a circle extends to an action of the algebra $C_\ast(S^1)$ of singular chains on $S^1$. This is because two elements of $S^1$ which are on the boundary of a $1$-chain in $S^1$ act homotopically.  We can pass to homology, and find that the algebra $H_\ast(S^1) = \C[\eps]$ (where $\eps$ is of degree $-1$) acts on the space of states. An action of this algebra is given by an operator of cohomological degree $-1$ which squares to zero.   

In the topological $B$-model, the space of states associated to a circle is $\PV^{\ast,\ast}(\C^5)$ with differential $\dbar$. The action of $H_\ast(S^1)$ is given by the operator $\partial$.

The moduli space of discs with a single marked point in the middle is of virtual dimension $-1$. The moduli space can be identified with the stack $BS^1 = \op{point} /S^1$.   We would like to understand the ``fundamental class'' of this space.  We will do this by thinking about the fundamental class of a space $M / S^1$ where $S^1$ acts freely on a manifold $M$.  

The fundamental chain of $M/S^1$ can be represented by a chain $\alpha \in C_{d-1}(M)$ such that $D \alpha = [M]$. Here $D$ is the operation which takes a $k$-chain and returns the $k+1$ chain swept out by it using the $S^1$ action. 

By analogy, we conclude that if $L$ is the linear map
$$
L : \PV_c^{\ast,\ast}(\C^5) \to \C
$$  
obtained from the disc with one point in the interior and boundary on the brane $\C^k$, then
$$
L ( \partial \alpha) = \int_{\C^k} \alpha \vee \Omega. 
$$ 
In other words,
$$
L(\alpha) = \int_{\C^k} (\partial^{-1} \alpha) \vee \Omega. 
$$

How does this relate to what happens in the physical string? The only possibility is that the field $\alpha \in \PV^{4-k,k}(\C^5)$ represents a component of the field-strength of the Ramond-Ramond $2k$-form.  Thus, $\partial^{-1} \alpha$ will be a component of the RR form, which is electrically coupled to the brane by integrating it over the brane.

  We further conclude that BCOV theory, as we have presented it, does not contain a field corresponding to the RR form itself.

Let us now use our calculation of the leading term in the action of BCOV theory in the presence of a $2k-1$ brane to calculate the equations of motion in the presence of the brane.  Let us calculate to leading order, by ignoring the interaction term in BCOV theory. Let us also use the original formulation of BCOV theory involving a non-local quadratic term.  The relevant term in the action is then
$$
\int_{\C^5} \alpha^{k,4-k} \wedge \dbar \partial^{-1} \alpha{4-k,k} + \int_{\C^k} (\partial^{-1} \alpha^{4-k,k}) \vee \Omega.
$$
This leads to the equation that
$$
\dbar \alpha^{k,4-k} = \delta_{\C^k}
$$
where distributional $(5-k,5-k)$ form $\delta_{C^k}$ is turned into a distributional $(k,5-k)$ poly-vector field.  

This tells us that the $D_{2k-1}$ brane is magnetically charged under $\alpha^{k,4-k}$.  

\subsection{Introducing RR forms into BCOV theory}
We can introduce new fields into BCOV theory whose field-strengths will be certain polyvector fields that appear in the original formulation.  In this section we will sketch briefly how to do this. 

As we have seen, in the physical string, the RR $k$-form $A_k$ and the RR $8-k$-form $A_{8-k}$ are not independent, but satisfy
$$
\ast \d A_k = \d A_{8-k}.  
$$
Thus, it is not possible to have a formulation where both the fields $A_k$ and $A_{8-k}$ are treated as fundamental fields. Instead, one has to treat half of the RR forms as fundamental, for example, $A_0,A_2,A_4$. Then, only the field-strengths of $A_6, A_8$ will exist in the theory, and will be represented by $\ast \d A_2, \ast \d A_0$.   

In type IIB there is a problem because the RR $4$-form is it's own electro-magnetic dual.  Thus, one of the constraints of the theory is that $\d A_4$ is self-dual.  One can not,therefore, have formulation of type IIB as a conventional field theory where $A_4$ is treated as a fundamental field. Either one introduces the constraint (not coming from the Lagrangian) that $\d A_4$ is self-dual, or one treats the self-dual $5$-form as a fundamental field and with the additional constraint that it is closed. 

We will find that exactly this phenomenon occurs in  BCOV theory. The way we have presented it, we have treated the RR field-strengths as the fundamental fields and had the constraint that the fields are in the kernel of the operator $\partial$.  (In our formulation, we represented this by setting up the theory as a degenerate theory in the BV formalism, which presumably one can also do for physical type IIB supergravity).  

However, we can introduce a different formulation of BCOV theory where some, but not all, of the original polyvector fields will be represented as $\partial$ of something. Let us see how this works. Let's introduce a new formulation of BCOV theory where we have fields in $\PV^{0,\ast}(\C^5)$, $\PV^{1,\ast}(\C^5)$ and $\PV^{2,\ast}(\C^5)$ exactly as before. But, now we introduce a new field say  $\alpha_2 \in  \PV^{2,\ast}(\C^5)$ such that $\partial \alpha_2$ is the field in $\PV^{3,\ast}(\C^5)$ we had before. Similarly, we introduce a field  in $\alpha_3 \in \PV^{3,\ast}(\C^5)$ such that $\partial \alpha_3$ is the field in $\PV^{4,\ast}(\C^5)$ we had before. (We will ignore $\PV^{5,\ast}$ because these fields do not propagate).  The fields $\alpha_2,\alpha_3$ have gauge symmetry: they are defined modulo $\partial$ of some element of $\PV^{1,\ast}$ and $\PV^{2,\ast}$ respectively. Further, there is gauge symmetry for the gauge symmetry.  

Let us write out the full space of fields for this new formulation of BCOV, including all ghosts and antifields, etc.  We will only include the fields that can propagate.  We find the complex of field is a direct sum
$$
\left( \oplus_{i+j \le 2}  t^j \PV^{i,\ast}(X) \right)  \bigoplus \left( \oplus_{\substack{k \ge 0 \\ l-k \ge 2}} t^{-k}\PV^{l,\ast } (X) \right).
$$
In the first summand, the fields in $t^j \PV^{i,r}(X)$ are in degree $2j+i+r-2$ as before.  In the second summand, the fields in $t^{-k} \PV^{l,r}(X)$ are in degree $-2k + l+r-1$.  The differential is, as before, $\dbar + t \partial$, with the convention that in the second summand the operation of multiplying by $t$ on a field accompanied by $t^0$ yields $0$. Fields with negative powers of $t$ are gauge transformations for the Ramond-Ramond fields.  

The fields in the second summand are the RR forms, together with their gauge transformations.  There is a map from this new complex of fields to the original complex of fields of the form
\begin{align*} 
 \left( \oplus_{i+j \le 2}  t^j \PV^{i,\ast}(X) \right)  \bigoplus \left( \oplus_{l-k\ge 2} t^{-k}\PV^{l,\ast } (X) \right) & \to \oplus_{i+j \le 4}  t^j \PV^{i,\ast}(X)\\
 t^j \PV^{i,\ast}(X) \ni  \alpha  & \mapsto \alpha \in  t^j \PV^{i,\ast}(X) \\
  t^{-k}\PV^{l,\ast } (X) \ni \alpha & \mapsto \delta_{k=0} \partial \alpha \in  \PV^{l+1,\ast}(X) . 
\end{align*}
In the last row, we see how this map sends an element $\alpha \in \PV^{l,\ast}(X)$ where $l \ge 2$ (thought of as a Ramond-Ramond form field) to $\partial \alpha$, which we interpret as one of the components of the field strength.  

The kernel which defines the BV anti-bracket for BCOV theory lifts to an element in the tensor square of this new space of fields.  If we restrict the interaction $I$ from the original space of fields to this new space of fields, we find automatically a solution to the classical master equation.   The same holds at the quantum level, so that a quantization of BCOV theory in our original formulation leads to a quantization of BCOV theory with this modified space of fields.  

\section{Twisted type IIB on $AdS_5 \times S^5$}

According to the AdS-CFT correspondence, type IIB supergravity on $AdS_5 \times S^5$ is equivalent to $\mscr{N}=4$ super Yang-Mills on the $4$-sphere.   We will propose a conjectural formulation of a twist of type IIB supergravity on $AdS_5 \times S^5$.  

The symmetries of the AdS background form the supergroup $\mf{psu}(2,2 \mid 4)$.  This is the Lorentzian signature form of the group: in Euclidean, as usual, the spin representations that appear do not have a real form, so its better to use the complexified Lie algebra of supersymmetries. This is $\mf{psl}(4 \mid 4)$.  

The bosonic part of this super Lie algebra is $\mf{sl}(4) \oplus \mf{sl}(4)$.  One copy of $\mf{sl}(4)$ corresponds to rotations of $S^5$; from the point of view of $\mscr{N}=4$ supersymmetric gauge theory it is the complexified Lie algebra of the $R$-symmetry group. We will refer to this copy as $\mf{sl}(4)_{R}$.  The other copy of $\mf{sl}(4)$ is the complexified isometries of hyperbolic $5$-space $\mbb{H}^5$, or equivalently the conformal symmetries of the $\R^4$ living on the boundary of $\mbb{H}^5$. We will refer to this copy as $\mf{sl}(4)_C$.  

Let $V_C$ denote the fundamental representation of $\mf{sl}(4)_C$, and $V_R$ that of $\mf{sl}(4)_R$. We will view $\mf{psl}(4 \mid 4)$ as being the projective symmetries of $V_C \oplus \Pi V_R$.   The fermionic part of the symmetry Lie algebra $\mf{psl}(4 \mid 4)$ is $V_C^\ast \otimes V_R + V_R^\ast \otimes V_C$.  

The complexified Poincar\'{e} group 
$$\mf{iso}(4) = \left(\mf{sl}(2) \oplus \mf{sl}(2)\right) \ltimes \C^4 $$
sits inside $\mf{sl}(4)_C$.  We are interested in twisting by a supercharge which behaves well with respect to the action of this subalgebra.
\begin{lemma}
Let us write the decomposition $\mf{so}(4,\C) = \mf{sl}(2)_+ \oplus \mf{sl}(2)_{-}$. Let us choose a copy of $\mf{sl}(3)$ inside $\mf{sl}(4)_C$. Then, there is a unique, up to rotation by $\mf{sl}(2)_+$,  fermionic element $Q \in \mf{psl}(4 \mid 4)$ which is invariant under the translation Lie algebra $\C^4$,  under $\mf{sl}(2)_{-} \subset \mf{sl}(4)_{C}$, and under $\mf{sl}(3) \subset \mf{sl}(4)_R$.  Up to conjugation in $\mf{psl}(4 \mid 4)$, $Q$ is characterized by the property that it is a matrix acting on $V_C \oplus \Pi V_R$ whose image is of dimension $(0 \mid 1)$.  

Further, $[Q,Q] = 0$ and the cohomology of $\mf{psl}(4 \mid 4)$ with respect to $Q$ is $\mf{psl}(3 \mid 3)$.   
\end{lemma}
\begin{proof}
As a representation of $\mf{so}(4,\C)$, the four-dimensional defining representation $V_C$ of $\mf{sl}(4)_C$ is decomposes as a direct sum $V_C = S_+ \oplus S_-$, where $S_{\pm}$ refer to the defining representations of the subalgebras $\mf{sl}(2)_{\pm} \subset \mf{so}(4)$.  It follows that as a representation of $\mf{so}(4) \oplus \mf{sl}(4)_R$, the space of fermionic symmetries decomposes as 
$$
V_C \otimes V_R^\ast \oplus V_C^\ast \otimes V_R = S_+ \otimes V_R \oplus S_- \otimes V_R \oplus S_+ \otimes V_R^\ast \oplus S_- \otimes V_R^\ast.$$
Here we are using the fact that $S_{\pm}$ are self-dual as representations of $\mf{sl}(2)_{\pm}$, so that $V_C$ is self-dual as a representation of $\mf{so}(4)$. 

The subspace of these fermionic symmetries that is invariant under the translations $\C^4 \subset \mf{sl}(4)_C$ is the $16$ dimensional subspace $S_+ \otimes V_C^\ast \oplus S_- \otimes V_C$.  (The fact that $S_+$ appears with $V_C^\ast$ and not $V_C$ is convention dependent: if we applied an inversion to $\mf{H}^5$, thus reversing the role of translations and the non-linear conformal symmetries of $\R^4$, we would find $S_+ \otimes V_C \oplus S_- \otimes V_C^\ast$). 

Now it is clear that any element $Q \in S_+ \otimes V_C^\ast$ which is invariant under our chosen $\mf{sl}(3) \subset \mf{sl}(4)_C$ is also invariant under translation and under $\mf{sl}(2)_-$.  Any such $Q$ can be decomposed as $Q = Q_0 \otimes v$ where $v \in V_C^\ast$ is invariant under $\mf{sl}(3)$ and $Q_0 \in S_+$ is arbitrary. This makes it clear that any two such $Q$'s are related by a rotation by $\mf{sl}(2)_+$.  

A simple calculation tells us that the cohomology of $\mf{psl}(4 \mid 4)$ with respect to $Q$ is $\mf{psl}(3 \mid 3)$. 
\end{proof}

Now, we can take a twist of type IIB supergravity in the AdS background, by putting it in the background where the bosonic ghost fields corresponding to the local supersymmetries are given constant value $Q$.

\subsection{The AdS background in BCOV theory}
We need to introduce the analog of the AdS background in BCOV theory. Our conjecture, of course, is that the twist discussed above of type IIB in the AdS background is equivalent to BCOV theory in the AdS background. 

In the physical AdS background, the only bosonic field that is non-zero (except for the metric) is the $5$-form Ramond-Ramond field-strength.  This is the field sourced by the $D3$ brane.   

We have already discussed how branes in the topological $B$-model can be magnetically coupled to the fields of BCOV theory, and so act as sources.  The natural guess is that, when we consider the twist of the AdS backgound, we need to introduce the field in BCOV theory which is sourced by a brane living on $\C^2 \subset \C^5$.  This field is an element of $\PV^{2,2}(\C^5 \setminus \C^2)$, which we can also think of as a $(3,2)$-form.  

Explicily, the $(3,2)$ sourced by a brane on $\C^2$ is the following.  Let 
$$
F \in \br{\Omega}^{3,2}(\C^5)
$$
be the unique $5$-form with tempered distributional coefficients which is harmonic outside $\C^2$, and satisfies
$$
\dbar F = \delta_{\C^2}
$$
where $\delta_{\C^2} \in \Omega^{3,3}(\C^5)$ is the de Rham current for the delta function on $\C^2$.

We can write $F$ in coordinates. Let $w_1, w_2, z_1, z_2, z_3$ be the coordinates of $\C^5$ such that our $\C^2$ is described by $z_1=z_2=z_3=0$. Let $r=\sqrt{|z_1|^2+|z_2|^2+|z_3|^2}$ be the radius of the direction normal to $\C^2$. Then  
$$
F =\frac{3}{4 i \pi^3} \d z_1 \d z_2 \d z_3 r^{-6} \left( \zbar_1 \d \zbar_2 \d \zbar_3 - \zbar_2 \d \zbar_1 \d \zbar_3 + \zbar_3 \d \zbar_1 \d \zbar_2  \right). 
$$
The normalization is such that
$$
\int_{ \sum \abs{z_i}^2 = 1}  F = 1
$$
which is necessary for the equation $\d F = \delta_{0}$, by Stoke's theorem. 

Let overload the notation $F$ and use it to indicate the polyvector field in $\PV^{2,2}(\C^5 \setminus \C^2)$ obtained from this $5$-form.   As a polyvector field, $F$ has the expression
$$
F = \frac{3}{4 i \pi^3} \dpa{w_1} \dpa{w_2} r^{-6}  \left( \zbar_1 \d \zbar_2 \d \zbar_3 - \zbar_2 \d \zbar_1 \d \zbar_3 + \zbar_3 \d \zbar_1 \d \zbar_2  \right). 
$$

We want to consider BCOV theory on $\C^5 \setminus \C^2$ in the background where the field in $\PV^{2,2}$ has value $F$, and all other fields are zero.  For this to make sense, we need to verify that $F$ satisfies the equations of motion.  The equations of motion for BCOV theory are that $F$ satisfies the Maurer-Cartan equation 
$$
\dbar F + t \del F + \tfrac{1}{2} \{F,F\} = 0.
$$
It is easy to see that each term in this equation vanishes. 
\begin{conjecture}
The twist of type IIB supergravity in the AdS background is BCOV theory on $\C^5 \setminus \C^2$ in the background where the field in $\PV^{2,2}$ has value $N F$. We are twisting with respect to the local supersymmetry $Q$ discussed above, and considering the version of the AdS background which is dual to $\mscr{N}=4$ Yang-Mills with gauge group $U(N)$. 
\end{conjecture}
The idea behind this conjecture is the following. Solving the equations of motion of supergravity when we place a $D3$ brane on $\R^4 \subset \R^{10}$ gives a metric with singularities along $\R^4$ as well as a $5$-form flux.  The AdS background is obtained by analyzing the near horizon limit of this geometry.

  We expect that the  metric will not play a role in the twisted theory, so that twisted type IIB in the presence of a $D3$ brane will be BCOV theory on $\C^5 \setminus \C^2$ with the $(3,2)$-form we discussed above. This configuration is homogeneous under scaling of the normal direction to the $D3$ brane on $\C^2$ (although when we scale these directions we also rescale the holomorphic volume form on $\C^5$, which can be counteracted by a change in the string coupling constant).  Thus, passing to the near horizon limit has no effect.   

\subsection{Matching symmetries}
To provide some evidence for this conjecture, we will match the symmetries present in the physical theory with those in the twisted theory.
 
The symmetries of the $AdS_5$ background of type IIB form the Lie algebra $\mf{psl}(4 \mid 4)$ (after complexifying). As we have seen, the $Q$-cohomology of $\mf{psl}(4 \mid 4)$ with respect to the $Q$ we have chosen is $\mf{psl}(3 \mid 3)$. We thus expect to see a copy of $\mf{psl}(3 \mid 3)$ living inside the dg Lie algebra of polyvector fields.  As we have seen, the $Q$-cohomology of $\mf{psl}(4 \mid 4)$ with respect to the $Q$ we have chosen is $\mf{psl}(3 \mid 3)$. We thus expect to see this $Q$-cohomology appearing in the dg Lie algebra describing BCOV theory in this background.

Recall that the dg Lie algebra describing BCOV theory on a Calabi-Yau $X$ is $\PV^{\ast,\ast}(X)\left\llbracket t\right\rrbracket[1]$, with differential $\dbar + t \del$ and with the Schouten Lie bracket.  If we choose a background field $\alpha \in \PV^{\ast,\ast}(X)$, then we change the differential to $\dbar + t \del + \{\alpha,-\}$.  This only makes sense when $F$ is an odd element of the Lie algebra satisfying the Maurer-Cartan equation.

If $\alpha$ is a degree $1$ element, then the new dg Lie algebra obtained by adding $\{\alpha,-\}$ to the differential continues to be $\Z$-graded.  If $\alpha$ is odd, but not of degree $1$, then we find a $\Z/2$ graded Lie algebra.

In the case of interest, 
$$\alpha = N F \in \PV^{2,2} (\C^5 \setminus \C^2)$$ 
where $F$ is as above.  As we have seen, $F$ satisfies the Maurer-Cartan equation.  However, $F$ of degree $3$ in the dg Lie algebra, and so the new deformed Lie algebra will be a $\Z/2$ graded dg Lie algebra.

Let us explicitly write down a copy of $\mf{psl}(3 \mid 3)$ in this dg Lie algebra. We will start with the case $N = 0$, so the differential is only $\dbar + t \del$, and then see that we can lift this copy of $\mf{sl}(3 \mid 3)$  to the deformation when $N$ is non-zero.

To do this, let us choose coordinates $z_i$ on $\C^3$ and $w_i$ on $\C^2$, as above.  We will assign to every element of $\mf{psl}(3 \mid 3)$ a polyvector field on $\C^5 \setminus \C^2$ which is in the kernel of the operator $\dbar + t \del + \{F,-\}$, and the Schouten bracket of these polyvector field will match the Lie bracket on $\mf{psl}(3 \mid 3)$.  

One of the two bosonic copies of $\mf{sl}(3)$ inside $\mf{psl}(3 \mid 3)$ is the one corresponding to $R$-symmetry in the $4d$ gauge theory.  In BCOV theory, this copy of $\mf{sl}(3)$ is given by the vector fields rotating $\C^3$, of the form $\sum A_{ij} z_i \dpa{z_j}$ the matrix $A_{ij}$ is in $\mf{sl}(3)$.  Note that these vector fields are holomorphic and divergence free.

The other bosonic copy of $\mf{sl}(3)$ corresponds to conformal transformations of $\C^2$, by which we mean holomorphic vector fields on $\C^2$ which extend to holomorphic vector fields on $\mbb{CP}^2$.  Such vector fields are given by:
\begin{enumerate}
\item Rotations of $\C^2$ by $\mf{sl}(2)$. A rotation by $A \in \mf{sl}(2)$ corresponds to the vector field $A_{ij} w_i \dpa{w_j}$ on $\C^5 \setminus \C^2$. 
\item Translations of $\C^2$. These correspond to the vector fields $\dpa{w_i}$ on $\C^5 \setminus \C^2$.
\item Scaling of $\C^2$.  This corresponds to the vector field
$$
\sum w_i \dpa{w_i} - \frac{2}{3} \sum z_i \dpa{z_i}.
$$
Note that this vector field on $\C^5 \setminus \C^2$ is divergence free.
\item Special conformal transformations.  These are the vector fields
$$
w_i \left( \sum_j w_j \dpa{w_j} - \sum_{k } z_k \dpa{z_k} \right)
$$
on $\C^5 \setminus \C^2$, for $i = 1,2$.   Note that these vector fields are divergence free.
\end{enumerate}
Next, let us describe the other copy of $\mf{sl}(3)$. This is simply by rotation in $\C^3$, so the vector fields are given by the formula $A_{ij} w_i \dpa{w_j}$ where $A_{ij} \in \mf{sl}(2)$.

Let us next write down the $18$ odd (fermionic) elements when $N = 0$. These are:
\begin{enumerate} 
 \item $z_i \in \PV^{0}(\C^5 \setminus \C^2)$.  These $3$ elements are ordinary supersymmetries (as opposed to superconformal symmetries) in the dual gauge theory.
 \item  $\dpa{w_i} \dpa{z_j} \in \PV^{2,0}(\C^5 \setminus \C^2)$.  These $9$ elements correspond again to ordinary supersymmetries. 
 \item The $6$ elements $z_i w_j \in \PV^{0}(\C^5 \setminus \C^2)$. These are superconformal symmetries in the dual gauge theory.
 \item The $3$ elements 
$$\dpa{z_i} \left(\sum_l w_l \dpa{w_l} -  \sum_k z_k \dpa{z_k}\right) \in \PV^{2,0}(\C^5 \setminus \C^2).$$
Note that these bivectors are in the kernel of $\partial$. These correspond to superconformal symmetries in the dual gauge theory.
\end{enumerate}
\begin{lemma}
The $16$ even and $18$ odd  elements we have written down in $\PV^{\ast}(\C^5 \setminus \C^2)$ form a copy of $\mf{psl}(3 \mid 3)$ inside the dg Lie algebra $\PV^{\ast}(\C^5\setminus \C^2)\left\llbracket t\right\rrbracket[1]$, with differential $\dbar + t \del$.  
\end{lemma}

\begin{proof}
This is easy to verify by an explicit calculation. 
\end{proof}

Of course, this is only the $N = 0$ case.  The more interesting case is when $N \neq 0$.  It turns out that not all of the polyvector fields we have written down commute with $N F$. However, we have the following.
\begin{proposition}\label{prop-NF}
The polyvector fields giving a copy of $\mf{psl}(3 \mid 3)$ have $N$-dependent corrections which make them  closed under $\dbar + t \del + \{F,-\}$.  Further, inside the cohomology of $\PV(\C^5 \setminus \C^2)\left\llbracket t\right\rrbracket [1]$ with differential $\dbar + t \del + N \{F,-\}$, these polyvector fields again form a copy of $\mf{psl}(3 \mid 3)$. 
\end{proposition}
Since this is a little technical, we will place the proof in Appendix \ref{appendix-B}.

\section{Twisted type IIA supergravity}

We will analyze the possible twists of type IIA supergravity, and give a conjectural description for some of them.  We will find that although there are $SU(5)$-invariant twists of type IIA, we do not have a candidate description, but only a description of certain $SU(4)$-invariant twists.  Our conjecture satisfies a number of checks. Our candidate for twisted type IIA has the correct residual supersymmetry; it is related to a twist of the theories on $D$ branes; it is $T$-dual to our candidate twist for type IIB; and it becomes the same as type IIB upon reduction to $8$ dimensions. We will prove that twisted type IIA can be quantized in perturbation theory, using a variant of the results of \cite{CosLi15}. 

\subsection{The twist}
Let us start by describing the twist we will consider. We will continue to use the notation introduced in section \ref{section-SUSY} for representations of $\mf{so}(10,\C)$: that is, the vector representation is $V$ and the two spin representations are $S_+$ and $S_-$. The supertranslation algebra algebra for type IIA supergravity is 
$$
\mc{T}^{(1,1)} = V \oplus \Pi ( S_+ \oplus S_- ). 
$$
The spinors $S_+$ and $S_-$ commute with each other, and the commutator on each space of spinors is given by the map
$$
\Gamma : S_{\pm} \otimes S_{\pm} \to V. 
$$
This Lie algebra has a central extension by the space $\Omega^\ast(\R^{10})[1]$ corresponding to the fundamental string. There is more than one way to represent this central extension (although different ways to represent it are equivalent). We will choose the cocycle defined by the map
$$
\Gamma_{\Omega^1} : S_- \otimes S_- \to \Omega^1 \subset \Omega^\ast(\R^{10})[1]. 
$$
We will call this centrally-extended algebra $\mf{superstring}_{IIA}$. 

Next, let us discuss the twist we will use.  This twist will be $\mf{sl}(4)$ invariant and not $\mf{sl}(5)$ invariant. To describe the supercharge we use, let us decompose $S_+$ and $S_-$ as representations of $\mf{sl}(5)$, by
\begin{align*} 
 S_+ &= \Omega^{0,ev}_{const}\\
 S_- &= \Omega^{0,odd}_{const}.
\end{align*}
The supercharge we choose is 
$$
Q = 1 + \d \zbar_1 \in \Omega^{0,0}_{const} \oplus \Omega^{0,1}_{const} \subset S_+ \oplus S_-. $$
This supercharge is square zero, even once we include the central extension.

\subsection{A conjectural description of the twist}
Our conjecture is the following.
\begin{conjecture}
This twist of the type IIA string is represented by a topological string on $\R^2 \times \C^4$, which in the $\R^2$ direction is the topological $A$-string and in the $\C^4$-direction is the topological $B$-string. 
\end{conjecture}
Of course, our main focus is not the twist of the string theory but of the supergravity theory.  To describe this we need to understand the closed-string field theory associated to the low-energy limit of the topological $A$-string. 

The topological $A$-model is described by Gromov-Witten theory. In the field theory limit, we will discard all instanton contributions, and only consider the Gromov-Witten theory of constant maps to the target Calabi-Yau $X$.  In this limit, the closed-string states is the de Rham cohomology $H^\ast_{dR}(X)$.  The three-point function for three  closed string states is simply $\int_X \alpha \wedge \beta \wedge\gamma$. 

This tells us that the closed-string field theory associated to the topological $A$-model is a topological theory on the target, unlike the closed-string field theory from the $B$-model.  After all, the de Rham cohomology of $X$, and the operations of wedging and integrating forms, are invariant under diffeomorphisms of $X$.

How would one describe the closed-string field theory associated to a mixture of the topological $A$ and $B$-models?  Let us propose an answer in the case of interest, when $X = \R^2 \times \C^4$.   The space of states of the string is the complex
$$
\Omega^\ast(\R^{2}) \what{\otimes} \PV(\C^4)
$$  
where we use the notation $\what{\otimes}$ to indicate the completed projective tensor product. This is essentially a notational shorthand for saying that the space of states is the sections of a bundle on $\R^2 \times \C^4$ which is the exterior algebra of $T^\ast \R^{2} \oplus (T^\ast)^{0,1} \C^4 \oplus T^{1,0} \C^4$. 

The differential on the space of states of the string is $\d_{dR}^{\R^2} + \dbar^{\C^4}$, a sum of the de Rham operator on $\Omega^\ast(\R^{2})$ and the Dolbeault operator on $\C^4$. 
As is usual in topological string theory, the space of states of the string has an action of the homology algebra $H_\ast(S^1)$ of the group $S^1$, coming from rotation of the string. The non-identity element in $H_1(S^1)$ acts by the operator $\partial$ on $\PV(\C^4)$.  The $S^1$-equivariant states of the string -- which become the fields of the corresponding closed-string field theory -- are 
$$
\Omega^\ast(\R^{2}) \what{\otimes} \PV(\C^4)\left\llbracket t\right\rrbracket 
$$
with differential $\d_{dR}^{\R^2} + \dbar^{\C^4} + t \partial^{\C^4}$.  

The fields of the closed string field theory are these equivariant states, with a shift of $[2]$. The theory is, as with the other versions of BCOV theory we consider, a degenerate theory in the BV formalism.  Thus we need to specify the kernel for the odd Poisson structure. This kernel is
$$
\pi = (\partial^{\C^4} \otimes 1) \delta_{Diag}
$$
where the delta function $\delta_{Diag}$ is a form with distributional coefficients on $(\R^2 \times \C^4)^2$, and we are using the identification
$$
\Omega^\ast(\R^2) \what{\otimes} \PV(\C^4) \iso \Omega^\ast(\R^2) \what{\otimes}\Omega^{\ast,\ast}(\C^4)
$$
provided by the holomorphic volume form on $\C^4$ to interpret $\delta_{Diag}$ as being an element of the tensor square of the space of fields.

There is an integration map $\op{Tr} : \PV_c(\C^4) \to \C$ which is zero on $\PV_c^{i,j}$ unless $i = j = 4$, and on $\PV^{4,4}_c(\C^4)$ is simply integration, after we have used the holomorphic volume form to make the identifications
$$
\PV^{4,4}(\C^4) = \Omega^{0,4}(\C^4) = \Omega^{4,4}(\C^4).
$$
If 
$$
\alpha = \sum \alpha_k t^k \in \Omega^\ast(\R^2) \what{\otimes} \PV(\C^4)\left\llbracket t\right\rrbracket
$$
is a field of the theory, then the interaction is of the form
$$
I(\alpha) = \int \alpha_0^3 + \text{ higher order terms.} 
$$
\begin{conjecture}
This field theory describes the $\mf{sl}(4)$-equivariant twist of type IIA supergravity.
\end{conjecture}

\subsection{Evidence for the conjecture}
Let us now summarize the evidence for the conjecture, which we will examine in more detail in subsequent sections.  The evidence is similar to that we presented for type IIB. Recall that we conjecture that the twist of the type IIA superstring is the topological string theory on $\R^2 \times \C^4$ which is the $A$-model on $\R^2$ and the $B$-model on $\C^4$.   Branes in this theory are products of $A$-branes on $\R^2$ and $B$-branes on $\C^4$. Since $A$-branes are Lagrangian submanifolds, these branes are given by submanifolds of $\R^2 \times \C^4$ of the form $\R \times \C^k$ for some $k \le 4$.   We will show that the theory living on the $\R \times \C^k$ brane is a twist of the theory living on a $D_{2k}$ brane in physical type IIA. 

We will also analyze the cohomology of the $10$-dimensional $(1,1)$ supersymmetry algebra, and show that this cohomology appears in the fields of our candidate twist of type IIA, just like we did for type IIB. We will also show that the residual supersymmetries present in our twist of type IIA which preserve a given brane are precisely the residual supersymmetries of the theory living on the brane.

Finally, we will show that our candidate for twisted type IIA string theory is $T$-dual to our candidate twist for type IIB, and that the supergravity theories become the same upon dimensional reduction to $9$ dimensions. 

\section{Residual supersymmetry in type IIA}
Recall that $\mf{superstring}_{IIA}$ is the central extension of the ten-dimensional $(1,1)$ supersymmetry algebra
$$
\mf{siso}_{IIA} = \mf{so}(10,\C) \ltimes \left( V \oplus \Pi S_+ \oplus \Pi S_-  \right) 
$$  
by the de Rham complex $\Omega^\ast(\R^{10})[1]$. Let $Q$ be the $\mf{sl}(4)$-invariant supercharge in $\mf{siso}_{IIA}$ discussed above.

We will calculate the Lie algebra   $\mf{superstring}_{IIA}^{Q}$, which is the cohomology of $\mf{superstring}_{IIA}$ with respect to the sum of the de Rham differential and the operator $[Q,-]$.  To do this, we need some notation. As above, let us decompose $\C^{5} = \R^2 \times \C^4$, and let $W = \C^4$ denote the fundamental representation of the $\mf{sl}(4)$ which preserves our chosen supercharge. 
\begin{lemma}
As an $\mf{sl}(4)$ representation, $\mf{superstring}^Q_{IIA}$ decomposes as 
$$
\mf{sl}(4)  \oplus W \oplus \wedge^2 W^\vee  \oplus \Pi \left( W^\vee \oplus \wedge^2 W^\vee \oplus \C \cdot c    \right) 
$$ 
where $c$ is the central element. 

The non-zero commutators between the fermionic elements are given by the map
$$
W^\vee \otimes \wedge^2 W^\vee \xto{\wedge} \wedge^3 W^\vee = W. 
$$
The commutator of $\mf{sl}(4)$ with anything is given by the natural $\mf{sl}(4)$ action. The only remaining commutator is between $W$ and $\Pi W^\vee$ which pair to give the central term $\Pi (\C\cdot c)$. 

\end{lemma}
\begin{proof}
Recall that we can identify
\begin{align*} 
 S_+ &= \Omega^{0,ev}_{const} \\
S_- &= \Omega^{0,odd}_{const}.
\end{align*}
Let us introduce a basis $z_1,\dots,z_5$ for $\C^5$ where the $z_2,\dots,z_4$ are a basis for $W = \C^4$ and $z_1$ is a basis for the direction in which we have the topological $A$-model. Our supercharge is $Q = 1 + \d \zbar_1$.  

We will start by calculating the cohomology of $\mf{siso}_{IIA}$, that is, without taking into account the central extension.  

Let us first calculate the kernel of $[Q,-]$ in the space of spinors.  The only non-zero brackets involving the spinor we denote by $1 \in \Omega^{0,0}_{const} \subset S_+$ are given by the map
$$
[1,-] : \Omega^{0,4} \to \C^5 \oplus \br{\C}^5
$$ 
which sends a form
$$
\d \zbar_1 \wedge \dots \wedge \what{\d \zbar_i} \dots \wedge \d \zbar_5 \mapsto \dpa{\zbar_i}. 
$$
The only non-zero brackets involving the spinor $\d \zbar_1$ are the maps
\begin{align*} 
 [\d \zbar_1,-] : \Omega^{0,3}_{const} & \to \br{\C}^5\\
 \d \zbar_2\wedge \dots \wedge \what{\d \zbar_i} \dots\wedge \d \zbar_5 & \mapsto \dpa{\zbar_i} \\
[\d \zbar_1,- ] : \Omega^{0,5}_{const} & \to \C^5 \\
\d \zbar_1\wedge \dots \wedge \d \zbar_5 & \mapsto \dpa{z_1}. 
\end{align*} 
Thus, the image of $[Q,-]$ in the space of vectors $\C^5 \oplus \br{\C}^5$ consists of $\dpa{z_1}$ and $\dpa{\zbar_i}$ for $i = 1,\dots,5$. This tells us that the twist is topological in the $z_1$ plane and holomorphic in the four other directions. 

The kernel of $[Q,-]$ therefore consists of $\Omega^{0,2}_{const}$, $\Omega^{0,1}_{const}$, the image of $\d \zbar_1 \wedge$ inside $\Omega^{0,3}_{const}$, and a subspace
$$
\br{\C}^4 = W^\vee \subset \Omega^{0,4}_{const} \oplus \Omega^{0,3}_{const}
$$
spanned by the elements 
$$
\dpa{z_i} \vee \left( \d \zbar_1 \wedge \dots \wedge \d \zbar_5 - \d \zbar_2 \wedge \dots \wedge \d \zbar_5  \right)
$$
for $i = 2,\dots,4$. 

The cokernel of $[Q,-]$ in $\br{\C}^5 \oplus \br{\C}^5$ consists, of course, of $W = \C^4$ spanned by $\dpa{z_2}, \dots,\dpa{z_4}$.  

Next let us analyze how $Q$ behaves under rotation by $\mf{so}(10,\C)$.  We can decompose
$$
\mf{so}(10,\C) = \mf{sl}(5,\C) \oplus \wedge^2 \C^5 \oplus \wedge^2 \br{\C}^5
$$
where $\mf{sl}(5,\C)$ acts in the evident way, $\wedge^2 \C^5$ acts on forms by multiplication with $\d \zbar_i \d \zbar_j$, and $\wedge^2 \br{\C}^5$ acts by contraction with $\dpa{\zbar_i} \dpa{\zbar_j}$.

The image of the map 
$$
[Q,- ] : \mf{so}(10,\C) \to \Omega^{0,\ast}_{const} = S_+ \oplus S_-
$$
consists of $\op{Im} \d \zbar_1 \wedge$ inside $\Omega^{0,2}_{const}$,  $\Omega^{0,1}_{const}$, and a subspace
$$
\wedge^2 \C^4 = \wedge^2 W \subset \Omega^{0,2}_{const} \oplus \Omega^{0,3}_{const}
$$
consisting of elements of the form 
$$\d \zbar_i \wedge \d \zbar_j - \d \zbar_1 \d \zbar_i \d \zbar_j$$
for $1 < i,j \le 5$.   

From this it follows that the fermionic part of the cohomology of $\mf{siso}_{IIA}$ consists of $W^\vee \oplus \wedge^2 W$, and a basis of cochain representatives is given by
\begin{align*} 
 \dpa{z_i} \vee \left( \d \zbar_1 \wedge \dots \wedge \d \zbar_5 - \d \zbar_2 \wedge \dots \wedge \d \zbar_5  \right)& \text{ for } i = 1,\dots 4\\
\d \zbar_i \d \zbar_j  & \text{ for } 2 \le i < j \le 4. 
\end{align*}

Next let us calculate the stabilizer of $Q$ in $\mf{so}(10,\C)$. This has two parts. First, we have the stabilizer in $\mf{sl}(5)$, which is
$$
\op{Stab}_{\mf{sl}(5,\C)} (Q) = \mf{sl}(4,\C) \oplus W^\vee 
$$
where we have used the decomposition 
$$
\mf{sl}(5,\C) =  \mf{sl}(4,\C) \oplus W \oplus W^\vee
$$
into $\mf{sl}(4,\C)$ representations. 

Then, we have the rest of the stabilizer, which consists of $\wedge^2 \br{\C^5}$ inside $\mf{so}(10,\C)$ in the decomposition given above. Thus,
$$
\op{Stab}(Q) = \mf{sl}(4,\C) \oplus W^\vee \oplus \wedge^2 \br{\C^5} = \mf{sl}(4,\C) \oplus W^\vee \oplus W^\vee \oplus \wedge^2 W^\vee.
$$ 
This argument shows us that the $Q$-cohomology of $\mf{siso}_{IIA}$ consists (as an $\mf{sl}(4)$ representation) fermionic elements $W^\vee \oplus \wedge^2 W$, translations $W$, and $\op{Stab}(Q)$ inside rotations. 

From the representation given above, it is easy to verify that the only non-zero bracket between fermionic elements is the composition
$$
W^\vee \otimes \wedge^2 W^\vee \to \wedge^3 W^\vee = W 
$$
landing in translations.  The $\mf{sl}(4)$ inside the stabilizer $\op{Stab}(Q)$ commutes with everything in the evident way. The two copies of $W^\vee$ inside the stabilizer bracket with each other via the map
$$
W^\vee \otimes W^\vee \to \wedge^2 W^\vee \subset \op{Stab}(Q). 
$$ inside rotations. 

The final non-zero brackets are that each copy of $W^\vee$ inside $\op{Stab}(Q)$ bracket with the copy of $W^\vee$ inside the fermions to give an element of $\wedge^2 W^\vee$ inside the fermions.

Finally, let us calculate the cohomology of the central extension. Using an argument similar to that given for the type IIB supersymmetry algebra, we find that the commutator between a fermion in $w^\vee \in W^\vee$ and a translation in $w \in W$ is $c \ip{w^\vee,w}$ where $c$ is the central element.  
\end{proof}

Now let us represent this supersymmetry algebra $\mf{superstring}_{IIA}^{Q}$ in the fields of our conjectural description of twisted type IIA supergravity.  Recall that the space of fields is $\Omega^\ast(\R^{2}) \what{\otimes}\PV(\C^4)\left\llbracket t\right\rrbracket[1]$.  The representation is as follows:
\begin{enumerate} 
\item Translations in $W= \C^4$ map to the vector fields $\dpa{z_i}$.  
\item Rotations $A \in \mf{sl}(4)$ map to the vector fields $\sum A_{ij} z_i \dpa{z_j}$. 
\item Fermions in $W^\vee$ map to the linear superpotentials $z_i$ for $i = 2,\dots, 5$. 
\item Fermions in $\wedge^2 W$ map to the bivectors $\dpa{z_i} \dpa{z_j}$ for $2 \le i < j \le 5$. 
\item  Rotations in one of the copies of $W^\vee$ map to the tri-vectors $\dpa{z_i} \dpa{z_j} \dpa{z_j}$. 
\item Rotations in the other copy of $W^\vee$ and in $\wedge^2 W^\vee$ map to zero.
\item The central element $c$ maps to the polyvector field $1 \in \PV^0$.  
\end{enumerate}
\begin{lemma}
This defines a homomorphism of super Lie algebras from
$$
\mf{superstring}_{IIA}^{Q} \to \Omega^\ast(\R^2) \what{\otimes} \PV(\C^4)[1]\left\llbracket t\right\rrbracket 
$$
where on the right hand side we use the Schouten  bracket. 
\end{lemma} 
\begin{proof}
It is easy to verify that the relevant commutation relations hold. 
\end{proof}

\section{Further twists of type IIA supergravity}\label{Further_twist_IIA}
Further twists of type IIA supergravity can be realized by setting some of the supersymmetries in the fields of the twisted theory to a non-zero value.  As in type IIB, there are two classes of further twists one can consider: we can make some of the four holomorphic directions non-commutative by setting some bivector $\dpa{z_i} \dpa{z_j}$ to a non-zero value, or we can turn on a linear superpotential $z_i$.  As in type IIB, turning on a linear superpotential makes the theory trivial in perturbation theory, and seems to lead to a theory where the partition function is given by ``counting'' gravitational instantons. Since we don't understand this very well right now, we will focus on the other kind of twist. 

If we turn on the bivector $\dpa{z_2} \dpa{z_3}$ in our conjectural twist of type IIA, we find a theory on $\R^6 \times \C^2$ which looks like the topological $A$-model on $\R^6$ and the topological $B$-model on $\C^2$. At the level of supergravity, the fields are
$$
\Omega^\ast(\R^6) \what{\otimes} \PV(\C^2)\left\llbracket t\right\rrbracket [2]
$$ 
with differential $\dbar + t \partial$.    It is natural to conjecture that, at the level of the string theory, we find something which is the topological $A$-model string in $6$ directions and the topological $B$-string in $2$ complex directions.  

Finally, let us discuss the maximally topological twist of type IIA. We can, of course, make all of $\C^4$ non-commutative by introducing the bivector $\dpa{z_1} \dpa{z_2} + \dpa{z_3} \dpa{z_4}$. The result is a theory where the fields are just $\Omega^\ast(\R^{10})\left\llbracket t\right\rrbracket$ with differential just the de Rham operator. At the level of supergravity, this theory is trivial. However, it is natural to guess that if we perform this twist at the level of the string theory we find the topological $A$-string on $\R^{10}$. 

\section{$T$-duality with type IIB and reduction to $9$ dimensions}\label{section-reduction}
We need to present some further evidence that our conjectural description of the twist of type IIA string t heory and supergravity is correct.   It is known that type IIA and type IIB superstring theories become equivalent (by $T$-duality) when compactified along a circle. In our situation, this becomes immediately clear.  If we compactify type IIA, we find a theory on $\R \times S^1\times \C^4$ which is the $A$-model on $\R \times S^1$ and the $B$-model on $\C^4$.  If we compactify type IIB, we find the $B$-model on $\C^\times \times \C^4$.  

It is a standard result in mirror symmetry (see \cite{AboAurEfiKat13}) that the topological $A$-model on $\R \times S^1$ is equivalent to the $B$-model on $\C^\times$. For example, the space of closed-string states in the $A$-model string can be identified with the space $\C[z,z^{-1}, \dpa{z}]$ of polyvector fields on $\C^\times$. We can see this as follows.  Closed-string states are computed as symplectic cohomology, which is the Floer cohomology theory describing Morse theory on the loop space.  Because $\R \times S^1$ is the cotangent bundle of $S^1$, the symplectic cohomology is equivalent \cite{Vit99, Abo13} to the homology of the free loop space of $S^1$.  The loop space of $S^1$ is homotopy equivalent to $\Z \times S^1$, where $\Z$ describes the winding number and $S^1$ describes the starting point of the loop in $S^1$.  The homology of this is, of course, $\C[z,z^{-1}, \eps]$ where $\eps$ is an odd parameter corresponding to a basis element of $H_1(S^1)$. The element $\eps$ maps to the element $\dpa{z}$ in polyvector fields on $\C^\times$.

As well as having an equivalence of string theories, we would like to have an equivalence of supergravity theories between the reductions of type IIA and type IIB on a circle.   Our candidate for twisted type IIA has fields $\Omega^\ast(\R^2) \what{\otimes} \PV(\C^4)\left\llbracket t\right\rrbracket[2]$. If we replace $\R^2$ by $\R \times S^1$ and consider the theory as a $9$-dimensional theory, we simply replace the de Rham complex of $S^1$ by its cohomology. We find a theory whose fields are
$$
\Omega^\ast(\R) \what{\otimes} \PV(\C^4)\left\llbracket t\right\rrbracket [\eps] [2]
$$ 
where $\eps$ corresponds to the generator of $H^1(S^1)$. 

Let us see how this theory arises by dimensional reduction from BCOV theory on $\C^5$.  We will first discuss some generalities about dimensional reduction of a holomorphic theory on $\C$ along a circle, which we will then apply to our situation.  Suppose we have a field theory on $\C$ where the fields are of the form $\Omega^{0,\ast}(\C,V)$ where $V$ is a graded vector space. Suppose that the linearized BRST operator is a sum of the Dolbeault operator with some holomorphic differential operator which preserves Dolbeault degree.  Let us put the theory on $\C^\times$ which we think of as a cylinder with coordinates $(x,\theta)$ where $x$ is a coordinate on $\R$ and $\theta$ is a coordinate on $S^1$.  We will reduce along the $\theta$-circle to get a one-dimensional field theory.  This means that we will only consider those fields which are invariant under rotation of the $\theta$-circle. If we do this, the Dolbeault complex becomes $\cinfty(\R)[\d \zbar]$ where $z = x + i \theta$ is a holomorphic coordinate.  The Dolbeault differential $\d \zbar \dpa{\zbar}$ becomes the operator $\dpa{x}$. After identifying $\d \zbar$ with $\d x$, we find that the $S^1$-invariant part of the Dolbeault complex has become the de Rham complex on $\R$. 

This general argument tells us that a field theory on $\C$ whose fields are of the form $\Omega^{0,\ast}(\C,V)$ with linearized BRST operator given by $\dbar$ becomes, upon dimensional reduction along $S^1$, a field theory on $\R$ whose fields are $\Omega^\ast(\R,V)$ with linearized BRST operator $\d_{dR}$. 

Applied to our situation, this argument tells us that BCOV theory on $\C^\times$, reduced along the circle to give a one-dimensional theory, gives a theory with fields $\Omega^\ast(\R)[\eps]\left\llbracket t\right\rrbracket [1]$ and where the linearized BRST operator is just $\d_{dR}$.  The odd parameter $\eps$ comes from the polyvector field $\dpa{z}$.  One might expect a term of the form $t \dpa{\eps} \dpa{x}$ to appear in the linearized BRST operator, coming from the operator $t \partial$ appearing in BCOV theory. However, the operator $\dpa{x}$ is cohomologous to zero under the de Rham operator on $\Omega^\ast(\R)$, so this term does not arise.

Finally, we can apply this to see that the reduction of BCOV theory on $\C^5$ to $\R \times \C^4$ along a circle gives rise to a theory whose fields are $\Omega^\ast(\R) \what{\otimes} \Omega^{0,\ast}(\C^4)[\eps]\left\llbracket t\right\rrbracket[1]$, where the linearized BRST operator is $\dbar + t \partial$.  This is the same as what we got from reducing our conjectural description of twisted type IIA, as desired. One can further check that this dimensional reduction is compatible with the kernel for the BV Poissin bracket and the interaction in BCOV theory. 

\section{$D$-branes in type IIA} 
Type IIA has $D$-branes living on odd dimensional submanifolds.  In this section we will explain how to realize these branes in our twisted version of type IIA.  We will also verify, as we did for type IIB, that the gauge theory living on these branes is a certain twist of maximally supersymmetric gauge theory, and that the residual supersymmetry acting on these twisted theories arises from the residual supersymmetry in type IIA. 

We will consider the $\mf{sl}(4)$ invariant twist of type IIA, which is (according to our conjecture) the theory on $\R^2 \times \C^4$ which is the topological $A$-model on $\R^2$ and the $B$-model on $\C^4$.  Branes in the topological $A$-model are given by Lagrangian submanifolds; of course, every one-dimensional submanifold of $\R^2$ is Lagrangian.  This suggests that the natural branes live on submanifolds of the form $\R \times \C^k \subset \R^2 \times \C^4$, where, for simplicity, we think of $\R \subset \R^2$ and $\C^k \subset \C^4$ as linearly embedded submanifolds. 
 
In the $A$-model, the gauge theory on a brane has fields built from the Hom complexes in the Fukaya category. If we discard instanton contributions -- which we are doing because we are interested in the supergravity limit -- we can model Hom's in the Fukaya category from a Lagrangian $L$ to itself by the de Rham complex $\Omega^\ast(L)$ of $L$.  

Applying this reasoning to our situation, we would guess that the fields in the $D$-brane gauge theory for a brane living on $\R \times \C^k$ inside $\R^2 \times \C^4$ are
$$
\Omega^\ast(\R) \what{\otimes} \Omega^{0,\ast}(\C^k) [\eps_1,\dots,\eps_{4-k} ] \otimes \gl_N[1]
$$ 
where we consider a stack of $N$ branes. The odd variables $\eps_\alpha$ appear for the same reason they do in our consideration of branes in type IIB.

The action functional should be the Chern-Simons type action functional
$$
\int_{\R \times \C^{k \mid 4-k} } \d w_1 \dots \d w_k \d \eps_1 \dots \d \eps_{4-k} \left(\tfrac{1}{2} \op{Tr} ( A (\dbar^{\C^4} + \d_{dR}^{\R^2} ) A ) + \tfrac{1}{3} \op{Tr} A^3 \right).
$$
As usual, integration of a differential form of mixed degree means that we pick out the part of top degree and integrate that.   

We will call this theory Chern-Simons theory on $\R_{dR} \times \C^{k \mid 4-k}$, where the inclusion of $\R_{dR}$ indicates that we use the de Rham complex on this factor, whereas we use the Dolbeault complex on $\C^k$. 

It turns out that this guess is correct:
\begin{lemma}
The minimal twist of the maximally supersymmetric gauge theory in dimensions $2k+1$ is equivalent to holomorphic Chern-Simons theory on $\R_{dR} \times \C^{k \mid 4-k}$. Further, the residual supersymmetries are implemented by the vector fields $\dpa{\eps_{\alpha}}$ and $\eps_{\alpha} \dpa{w_i}$ on $\C^{k \mid 4-k}$.  
\end{lemma}
\begin{remark}
By the ``minimal'' twist, we mean the twist that is as close to the physical theory as possible, so that as few as possible directions are made topological.  The corresponding supercharge $Q$ is $\mf{sl}(k,\C)$ invariant and has the feature that the image of $[Q,-]$ consists of the smallest possible number of complexified translations, which in this case is $k+1$.  
\end{remark}
\begin{proof}
This twist is reduced form the holomorphic twist of the maximally supersymmetric gauge theory in dimension $2k+2$, which we have already seen is holomorphic Chern-Simons theory on $\C^{k+1 \mid 4-k}$.  The argument in section \ref{section-reduction} tells us that reducing this theory to $2k+1$ dimensions amounts to replacing one copy of $\C$ by $\R_{dR}$.  

We know that the residual supersymmetries for the $D_{2k+1}$ brane theory in type IIB are represented by the vector fields $\dpa{\eps_{\alpha}}$ and $\eps_\alpha \dpa{w_i}$ on $\C^{k+1 \mid 4-k}$. It follows that the vector fields $\dpa{\eps_{\alpha}}$ and $\eps_\alpha \dpa{w_i}$, where $i > 1$, which act on the $\C^{k \mid 4-k}$ factor of $\R_{dR}\times  \C^{k \mid 4-k}$, must be part of the residual supersymmetry algebra of the minimal twist of the $D_{2k}$-brane theory in type IIA.  A simple cohomology calculation tells us that there are no more residual supersymmetries. 
\end{proof}
This tells us that $D$-branes in  our twisted IIA supergravity, just like in type IIB, behave exactly as one expects from the usual physics story.  

\subsection{Coupling the supergravity theory to the theory on a $D$-brane}
One can show that the fields of our candidate for the twist of type IIA supergravity couple to the $D$-brane theory by an argument similar to the one we employed for type IIB.  The theory on a $D_{2k}$-brane is, as we have seen, Chern-Simons on $\R_{dR} \times \C^{k \mid 4-k}$.  One can show that the local Hochschild cohomology of $\Omega^\ast(\R) \what{\otimes} \Omega^{0,\ast}(\C^k) [\eps_{\alpha}]$ is 
$$
\Omega^\ast(\R) \what{\otimes} \PV^{\ast,\ast}(\C^{k \mid 4-k})
$$ 
where the polyvector fields on the complex supermanifold $\C^{k \mid 4-k}$ were discussed in section \ref{section-HH}. 

It follows that the local cyclic cohomology of the same algebra, which describes the universal single-trace deformations of the Chern-Simons action, is
$$
\Omega^{\ast}(\R) \what{\otimes} \PV^{\ast,\ast}(\C^{k \mid 4-k} )\llbracket t \rrbracket
$$ 
with a differential which includes a term $t \partial$ where $\partial$ is the divergence operator. 

In section \ref{section-HH} we explained how to write down a cochain map
$$
\PV^{\ast,\ast}(\C^5) \llbracket t \rrbracket \to \PV^{\ast,\ast}(\C^{k \mid 5-k})\llbracket t \rrbracket. 
$$
The same formula applies if we replace $5$ by $4$. Together with the pull-back map $\Omega^\ast(\R^2) \to \Omega^\ast(\R)$ this map gives us a cochain map
$$
\Omega^\ast(\R^2) \what{\otimes} \PV(\C^4)\llbracket t \rrbracket  \to \Omega^\ast(\R) \what{\otimes} \PV^{\C^{k \mid 4-k}} ) )\llbracket t \rrbracket  
$$
which implements the desired coupling.

As in type IIB, one can calculate that the residual supersymmetries of the theory on a $D$-brane arise from supersymmetries living in the fields of twisted type IIA supergravity by this map.

\appendix
 
\section{}\label{Appendix-A}
In this section we will prove theorems \ref{thm-IIB-SUSY}  showing that the residual supersymmetry algebra on the minimal twists of $D$-brane gauge theories in type IIB and type IIA acts in the way we described, and arises from the supersymmetries of twisted supergravity.  

\begin{theorem} 
 Consider the maximally supersymmetric gauge theory in dimension $2k$. Let $\mc{T}_{2k}$ denote the corresponding supertranslation algebra, which is of the form
$$
\mc{T}_{2k} = \C^{2k} \oplus \Pi S
$$ 
where $\C^{2k}$ is the complexification of the translations $\R^{2k}$ and $S$ is a $16$-dimensional spin representation of $\mf{so}(2k,\C)$ restricted from an irreducible spin representation of $\mf{so}(10,\C)$.  The super Lie algebra $\mc{T}_{2k}$ is acted on by space-time rotations $\mf{so}(2k,\C)$ and by the $R$-symmetry group $\mf{so}(10-2k,\C)$.

Let $Q \in S_{2k}$ be the unique up to scale element which is invariant under $\mf{sl}(5,\C) \subset \mf{so}(10,\C)$.  The supercharge $Q$ gives rise to the holomorphic twist.

Let 
$$
\mf{siso}^{R}(2k) = \left(\mf{so}(2k,\C) \oplus \mf{so}(10-2k,\C) \right) \ltimes \mc{T}_{2k}
$$
be the Lie algebra obtained by adding space-time rotations and the $R$-symmetry Lie algebra to the supertranslation algebra. 

Then, the $Q$-cohomology of $\mf{siso}^{R}(2k)$ acts on the twisted theory, which is holomorphic Chern-Simons on $\C^{k \mid 5-k}$.  This $Q$-cohomology consists of: 
\begin{enumerate} 
 \item The stabilizer 
$$\op{Stab}(Q) \subset \mf{so}(2k,\C) \oplus \mf{so}(10-2k,\C).$$
This is subalgebra includes $\mf{sl}(k,\C) \oplus \mf{sl}(5-k,\C)$ (and we will largely ignore elements of $\op{Stab}(Q)$ which are not in $\mf{sl}(k,\C) \oplus \mf{sl}(5-k,\C)$). The group $\mf{sl}(k,\C) \oplus \mf{sl}(5-k,\C)$ acts on the twisted theory via the obvious action on $\C^{k \mid 5-k}$, where the $k$ bosonic directions are in the fundamental representation of $\mf{sl}(k,\C)$ and the $5-k$ fermionic directions are in the anti-fundamental representation of $\mf{sl}(5-k,\C)$. 
\item Translations $\C^k$ in the holomorphic directions $\dpa{w_i}$ on $\R^{2k} = \C^k$. These act on the twisted theory by the vector fields $\dpa{w_i}$ of translation in the bosonic directions of $\C^{k \mid 5-k}$. 
\item $(5-k)(k+1)$ fermionic elements, living in a representation of $\mf{sl}(k,\C) \oplus \mf{sl}(5-k,\C)$ of the form $W^\vee \oplus W \otimes V$, where $W$ is the fundamental representation of $\mf{sl}(5-k,\C)$ and $V$ is the fundamental representation of $\mf{sl}(k,\C)$.   The commutator of a supercharge $w^\vee \in W^\vee$ with $w \otimes v \in W \otimes V$ is 
$$
[w^\vee, w \otimes v] = \ip{w^\vee,w} v
$$ 
where we view $v$ as a translation. 

These act on the twisted theory by the vector fields $\dpa{\eps_\alpha}$ and $\eps_\alpha \dpa{w_j}$ on $\C^{k \mid 5-k}$. 
\end{enumerate}
 
\end{theorem} 

\begin{remark}
We have not excluded the possibility that there might be some higher $L_\infty$ corrections that appear in the action of the residual supersymmetry algebra on the twisted theory. We believe, however, that this is not the case.  This is difficult to prove directly because of the lack of an off-shell formulation of maximally supersymmetric gauge theory in which supersymmetry is evident.  Probably one could prove this by a deformation-theory argument, showing that no such higher $L_\infty$ corrections are possible. 
\end{remark}
\begin{proof}
The proof of the theorem  is rather lengthy. The main steps are the following:
\begin{enumerate} 
 \item Compute the $Q$-cohomology of the supersymmetry algebra (including $R$-symmetry) in various dimensions and verify that the correct commutation relations hold.  This is very similar to the calculation we already performed for the $(2,0)$ supersymmetry algebra in $10$ dimensions. 
\item Check that this $Q$-cohomology acts in the way we expect on the twisted theory.  The proof of this is rather indirect. We know, from general arguments, that there must be such an action on the twisted theory and that translations, $R$-symmetry, and space-time rotations must act in the obvious way.  We then find that there's a unique way to make the fermionic elements act so that the commutation relations hold. 
\end{enumerate}

Let us start by computing the $Q$-cohomology of the supersymmetry algebra in various dimensions, starting with $k = 0$.  In this case, there are no translations in the supersymmetry algebra, so that every element in $S$ is in the kernel of $[Q,-]$.  The image of $[Q,-]$ consists of those elements of $S$ which can be obtained from $Q$ by rotation by an element of $\mf{so}(10,\C)$.  It is easy to verify (following our discussion in the case of the $(2,0)$ supersymmetry algebra in $10$ dimensions) that $\op{Im}[Q,-]$ is $11$ dimensional, so that the fermionic elements in $Q$-cohomology are a $5$ dimensional space, which forms the dual of the fundamental representation of $\mf{sl}(5,\C)$.

The case of $k = 5$ is entirely parallel to the case $k = 0$. In this case, $[Q,-]$ takes spinors to vector fields $\dpa{\br{w_i}}$ for $i = 1,\dots,5$. Thus, the kernel of $[Q,-]$ in $S$ is $11$ dimensional, and coincides with those spinors that can be obtained by rotating $Q$ by an element of $\mf{so}(10,\C)$. Therefore, in this case, there are no fermionic elements in the $Q$-cohomology.

Next, let us discuss the cases $k = 1$ and $k = 4$.  In this case, the $16$ dimensional spin representation $S$ decomposes as
$$
S = S_+^{(2)} \otimes S_+^{(8)} \oplus S_-^{(2)} \otimes S_-^{(8)}
$$  
where $S_{\pm}{(2)}$ are the two one-dimensional spin representations of $\mf{so}(2,\C)$ and $S_{\pm}^{(8)}$ are the two distinct eight-dimensional spin representations of $\mf{so}(8,\C)$. The representations $S_{\pm}^{(8)}$ are self-dual, whereas $S_{+}^{(2)}$ and $S_-^{(2)}$ are dual to each other.  The vector representation of $\mf{so}(8,\C)$ is a summand of $S_+^{(8)} \otimes S_-^{(8)}$, and the vector representation of $\mf{so}(2,\C)$ is $S_+^{(2)} \otimes S_+^{(2)} \oplus S_-^{(2)} \otimes S_-^{(2)}$.  The supercharge we use to twist is an element of $S_+^{(2)} \otimes S_+^{(8)}$ which is the tensor product of a basis element of $S_+^{(2)}$ with a null vector in $S_+^{(8)}$.    

In the case $k = 1$, so that we are discussing a $D1$ brane, the Lie bracket on the supersymmetry algebra is such that elements in $S_+^{(2)} \otimes S_+^{(8)}$ commute with those in $S_-^{(2)}\otimes S_-^{(8)}$. The kernel of $[Q,-]$ contains all of $S_-^{(2)} \otimes S_-^{(8)}$ and a seven dimensional subspace of $S_+^{(2)} \otimes S_+^{(8)}$.  The image of $Q$ under rotations by $\mf{so}(8,\C)$ is also seven-dimensional subspace of $S_+^{(2)} \otimes S_+^{(8)}$.      It follows that in this case, the fermionic elements in the $Q$-cohomology consist of all of $S_-^{(2)} \otimes S_-^{(8)}$.  Further, as a representation of $\mf{sl}(4,\C) \subset \mf{so}(8,\C)$, $S_-^{(8)}$ decomposes as a direct sum of two $4$-dimensional representations, which are the fundamental and anti-fundamental representations of $\mf{sl}(4,\C)$. Let $U$ denote the fundamental representation of $\mf{sl}(4,\C)$. Then, the fermions in the $Q$-cohomology are $U \oplus U^\vee$, and the commutator is such that $$[u^\vee,u] = \ip{u^\vee,u}\dpa{w}$$ 
is translation in the holomorphic direction in the space $\R^2 = \C$ where our $D_1$ brane lives.    Thus, the fermionic elements form the same representation of $\mf{sl}(4,\C)$ and have the same commutators as the vector fields $\dpa{\eps_\alpha}$ and $\eps_\alpha \dpa{w}$ on $\C^{1\mid 4}$.

Next let us focus on $k = 4$, so that we are discussing a $D7$ brane.  The $16$ supercharges form the same representation $S_+^{(2)} \otimes S_+^{(8)} \oplus S_-^{(2)} \otimes S_-^{(8)}$ of $\mf{so}(2,\C) \oplus \mf{so}(4,\C)$ and we are twisting with respect to the same $Q$ as in the $k = 1$ case.  In this case, the non-zero commutators in the supersymmetry algebra are between $S_+^{(2)} \otimes S_+^{(8)}$ and $S_-^{(2)} \otimes S_-^{(8)}$, where we use the fact that the vector representation of $\mf{so}(8,\C)$ is a summand of $S_+^{(8)} \otimes S_-^{(8)}$ and that the representations $S_+^{(2)}$ and $S_-^{(2)}$ are dual.  It follows that the kernel of $[Q,-]$ consists of all of $S_+^{(2)} \otimes S_+^{(8)}$ and a $4$-dimensional subspace of $S_-^{(2)} \otimes S_-^{(8)}$ which transforms under the fundamental representation of $\mf{sl}(4,\C)$.  The image of $[Q,-]$ in the space of fermions consists of a $7$-dimensional subspace of $S_+^{(2)} \otimes S_+^{(8)}$. Therefore, at the level of cohomology, we have $5$ supercharges, $4$ of which form the fundamental representation of $\mf{sl}(4,\C)$ and one of which is invariant under $\mf{sl}(4,\C)$.  These $5$ supercharges correspond to $ \dpa{\eps}$ and $\eps \dpa{w_i}$ for $i = 1,\dots,4$ (and have the same commutation relations).

The final two cases are $k = 2$ and $k = 3$.  In these cases, the $16$ dimensional spin representation $S$ is decomposed as a direct sum $S = S_+^{(4)} \otimes S_+^{(6)} \oplus S_-^{(4)} \otimes S_-^{(6)}$ where $S_{\pm}^{(4)}$ are the two irreducible spin representations of $\mf{so}(4,\C)$, each of rank $2$; and $S_{\pm}^{(6)}$ are the two rank $4$ spin representations of $\mf{so}(6,\C)$.  The representations $S_{\pm}^{(4)}$ are self dual, whereas $S_+^{(6)}$ is dual to $S_-^{(6)}$.  The vector representation of $\mf{so}(4,\C)$ is $S_+^{(4)} \otimes S_-^{(4)}$ and the vector representation of $\mf{so}(6,\C)$ is $\wedge^2 S_+^{(6)} = \wedge^2 S_-^{(6)}$.  The supercharge we choose is the tensor product of an element in $S_+^{(4)}$ with an element in $S_+^{(6)}$.  

In the case $k = 2$, the non-trivial commutators in the supersymmetry algebra are between $S_+^{(4)} \otimes S_+^{(6)}$ and $S_-^{(4)} \otimes S_-^{(6)}$.  The kernel of $[Q,-]$ thus consists of all of $S_+^{(4)} \otimes S_+^{(6)}$ and a $6$ dimensional subspace of $S_-^{(4)} \otimes S_-^{(6)}$.  This $6$-dimensional subspace transforms under $\mf{sl}(2,\C) \oplus \mf{sl}(3,\C)$ as the tensor product of the fundamental representations.  The image of $[Q,-]$ in the space of fermions is a $5$ dimensional subspace of $S_+^{(4)} \otimes S_+^{(6)}$. Therefore there are $9$ fermionic elements in the $Q$-cohomology, $6$ of whom form the tensor product of the fundamental representation of $\mf{sl}(2,\C)$ with that of $\mf{sl}(3,\C)$ and the remaining $3$ form the dual of the fundamental representation of $\mf{sl}(3,\C)$. Thus, as representations of $\mf{sl}(2,\C) \oplus \mf{sl}(3,\C)$, they behave in the same way as the vector fields $\eps_\alpha \dpa{w_j}$ and $\dpa{\eps_\alpha}$ on $\C^{2 \mid 3}$  we wrote down earlier. To check that the commutators are correct, note that there is a unique up to scale $\mf{sl}(2) \oplus \mf{sl}(3)$-invariant map from the symmetric square of our $9$-dimensional space of fermions in the $Q$-cohomology to the space of translations on $\C^2$, which form the fundamental representation of $\mf{sl}(2,\C)$.  Since we know the commutators at the level of $Q$-cohomology must be non-zero, this representation theory argument tells us that they must be correct.  

Finally, let us discuss the case $k = 3$, where the fermions form the same representation of $\mf{so}(4,\C) \oplus \mf{so}(6,\C)$ as in the case $k = 2$, and we use the same supercharge.  In this case, the non-zero commutators between fermions are between two elements of $S_+^{(4)} \otimes S_+^{(6)}$ and between two elements of $S_-^{(4)} \otimes S_-^{(6)}$. Therefore, the kernel of $[Q,-]$ in the space of fermions consists of all of $S_-^{(4)} \otimes S_-^{(6)}$ and a $5$ dimensional subspace of $S_+^{(4)} \otimes S_+^{(6)}$.  This $5$-dimensional subspace must be a subspace of the image of $Q$ under rotation by $\mf{so}(4,\C) \oplus \mf{so}(6,\C)$, and we have already seen that this image is $5$ dimensional, so they coincide. It follows that the $Q$-cohomology consists of all of $S_-^{(4)} \otimes S_-^{(6)}$. The space $S_-^{(4)}$ is the fundamental representation of $\mf{sl}(2,\C)$ (in this case fundamental and anti-fundamental coincide). The space  $S_-^{(6)}$ is a direct sum of the fundamental representation of $\mf{sl}(3,\C)$ with the trivial representation. It follows that the $8$ supercharges present in $Q$-cohomology transform in the same way under $\mf{sl}(2,\C) \oplus \mf{sl}(3,\C)$ as the vector fields $\dpa{\eps_\alpha}$, $\eps_\alpha \dpa{w_j}$ on $\C^{3 \mid 2}$ we wrote down earlier.  As in the case $k = 2$, one can check that the commutation relations we find in the $Q$-cohomology of the supersymmetry algebra must coincide with those among these vector fields on $\C^{3 \mid 2}$ using a representation theory argument.   

So far, we have checked that for $k = 0,\dots, 5$, the fermions in the $Q$-cohomology of the supersymmetry algebra of the maximally supersymmetric gauge theory on $\R^{2k}$ match with the vector fields we wrote down on $\C^{k \mid 5-k}$ at the level of representations of $\mf{sl}(k,\C) \oplus \mf{sl}(5-k,\C)$ and that the commutation relations also match.  It remains to verify that these fermions act on the space  of fields of the twisted theory (which is holomorphic Chern-Simons on $\C^{k \mid 5-k}$) by these vector fields on $\C^{k \mid 5-k}$.  

To prove this, we will start with the case $k = 0$.  In this case the field theory is holomorphic Chern-Simons theory on $\C^{0 \mid 5}$ and the residual supersymmetry consists of $5$ fermionic symmetries transforming in the dual of the fundamental representation of $\mf{sl}(5)$.  We will let $V$ denote the fundamental representation of $\mf{sl}(5)$. The supermanifold on which we are considering holomorphic Chern-Simons is $\Pi V^\ast$. 

The supersymmetries will act on $\gl_N$ holomorphic Chern-Simons theory on $\C^{0 \mid 5}$ for all $N$, in a way compatible with the inclusions from $\gl_N \into \gl_{N+k}$. Using the arguments presented in \cite{CosLi15}, this implies that the supersymmetries must be represented by a map to the cyclic cohomology groups of the exterior algebra on $V$, with a shift of one. This cyclic cohomology group is, as we have discussed earlier, the space of polyvector fields on $\pi V$, which we can identify (after introducing a shift)
$$
\wedge^\ast V \otimes \what{\Sym}^\ast V^\ast\left\llbracket t\right\rrbracket[1] = \C\left\llbracket\eps_\alpha, \dpa{\eps_\alpha}, t \right\rrbracket[1]
$$
with differential $t\partial$ where $\partial$ is the divergence operator. 

The cohomology of this complex is the subspace 
$$
\op{Ker} \partial \subset \wedge^\ast V \otimes \what{\Sym}^\ast V^\ast[1].
$$  
Thus, we need a Lie algebra homomorphism from the Abelian super Lie algebra $\pi V^\ast$ to this space, and in particular a Lie algebra homomorphism to the larger space $\wedge^\ast V \otimes \what{\Sym}^\ast V^\ast$.  This homomorphism must be $\mf{sl}(5)$-invariant.

Such a homomorphism is, in particular, an $\mf{sl}(5)$- invariant element in 
$$
V \otimes \wedge^\ast V \otimes \what{\Sym}^\ast V^\ast.
$$
Invariant theory for $\mf{sl}(5)$ tells us that the space of invariants is of rank two, with one invariant element in $V \otimes \wedge^4 V$ and the other in $V \otimes \Sym^1 V^\ast$.  In coordinates, the first invariant element corresponds to the supersymmetries acting by the polyvector fields $\eps_1 \dots \what{\eps_\alpha} \dots \eps_5$, and the other corresponds to the supersymmetries acting by $\dpa{\eps_\alpha}$. The first possibility can be excluded because these are bosonic elements in the super Lie algebra of polyvector fields.  It follows that we have proved that the only possible action of the residual supersymmetries in the case $k = 0$ is by the vector fields $\dpa{\eps_\alpha}$ acting on $\C^{0 \mid 5}$. 

Next, let us discuss the case $k > 0$.  We find, by an argument similar to that used in the $k = 0$ case,  that we need a homomorphism of Lie algebras from our residual supersymmetry algebra to the cyclic cohomology of $\Omega^{0,\ast}(\C^{k \mid 5-k})$, which is $\PV(\C^{k \mid 5-k})\left\llbracket t\right\rrbracket$.  At the level of cohomology, only polyvector fields that are of Dolbeault degree $0$ and in the kernel of $\dbar$ survive.  We find that the cohomology of the space of polyvector fields is
$$
\op{Ker} (\partial) \subset \op{Hol}(\C^k) \left \llbracket \dpa{w_i}, \eps_\alpha, \dpa{\eps_\beta} \right \rrbracket
$$
where $\op{Hol}(\C^k)$ is the space of holomorphic functions on $\C^k$ and $\partial$ is the divergence operator.  

Our residual supersymmetries will be represented by $(5-k)(k+1)$ even polyvector fields (even because of the parity change one needs to make polyvector fields into a Lie algebra).  Since, if we dimensionally reduce the theory to a point, the supersymmetries are represented by the vector fields $\dpa{\eps_\alpha}$ on $\C^{0 \mid 5}$, we conclude that $5-k$ of the desired supersymmetries are represented by the vector fields $\dpa{\eps_\alpha}$ on $\C^{k \mid 5-k}$.  The remaining $(5-k)k$ polyvector fields must commute with the $\dpa{\eps_\alpha}$ to give a translation $\dpa{w_j}$. Therefore, they must be of the form
$$
\eps_\alpha \dpa{w_j} + A_{\alpha j}
$$
where the $A_{\alpha j}$ commute with all $\dpa{\eps_\alpha}$ and all $\dpa{w_i}$ and transform in the tensor product of the fundamental representations of $\mf{sl}(k)$ and $\mf{sl}(5-k)$.  The $A_{\alpha j}$ can not involve any $\eps$'s, and as functions of $w \in \C^k$ they must be constant, so that they are in 
$$
\C \left \llbracket \dpa{w_i}, \dpa{\eps_\alpha} \right \rrbracket. 
$$ 
If $1 < k < 4$, then the only elements which transform in the correct representation are  of the form
$$
A_{\alpha j} = \dpa{w_i} \dpa{\eps_1} \dots \what{\dpa{\eps_\alpha}} \dots \dpa{\eps_{5-k}}.
$$
But these elements are odd polyvector fields, whereas to correspond to supersymmetries they must be even.  Therefore $A_{\alpha j} = 0$ in the cases where $1 < k < 4$.  The remaining cases are $k = 1$ and $k = 4$.  For $k = 4$, there is only one fermionic direction, so rotations by $\mf{sl}(5-k) = \mf{sl}(1)$ give no constraints. This implies that there is an extra possibility where $A_{1 j} = \dpa{w_i}$. This is again an odd polyvector field and so can not correspond to a supersymmetry. For $k = 1$, there is an extra possibility of the form 
$$A_{\alpha 1} =  \dpa{\eps_1} \dots \what{\dpa{\eps_\alpha}} \dots \dpa{\eps_{5-k}}.$$
This possibility can be excluded because in this case the supersymmetries are not dimensionally reduced from the case $k = 2$.  

Therefore, $A_{\alpha j} = 0$ and we find that the action of the residual supersymmetry on the holomorphic twist is what we claimed it is.  
\end{proof}

\section{}\label{appendix-B}
In this section, we will prove a technical proposition concerning the supersymmetries of BCOV theory in the AdS background, that we stated but did not prove earlier.   
\begin{proposition}
The polyvector fields giving a copy of $\mf{psl}(3 \mid 3)$ have $N$-dependent corrections which make them  closed under $\dbar + t \del + \{F,-\}$.  Further, inside the cohomology of $\PV(\C^5 \setminus \C^2)\left\llbracket t\right\rrbracket [1]$ with differential $\dbar + t \del + N \{F,-\}$, these polyvector fields again form a copy of $\mf{psl}(3 \mid 3)$. 
\end{proposition}

\begin{proof}
The operation $\{F,-\}$ maps $\PV^{i,j}$ to $\PV^{i+1,j+2}$.  There is a spectral sequence which computes the cohomology of $\PV(\C^5 \setminus \C^2)\left\llbracket t\right\rrbracket$ associated to the filtration where $F^k \PV(\C^5 \setminus \C^2)\left\llbracket t\right\rrbracket$ consists of those series $\sum t^n \alpha_n$ where $\alpha_n \in \PV^{\ge k-n, \ast}(\C^5 \setminus \C^2)\left\llbracket t\right\rrbracket$.  The operators $\dbar$, $t \partial$, and $\{F,-\}$ all preserve this filtration.  On the associated graded, the operator $\{F,-\}$ becomes zero, so that the first page of this spectral sequence is the cohomology of $\PV(\C^5 \setminus \C^2)\left\llbracket t\right\rrbracket$ with respect to the differential $\dbar + t \partial$.    

Let us compute the first page of this spectral sequence. It is convenient to use an additioanl spectral sequence, obtained from the filtration by powers of $t$, whose first page gives the cohomology of $\PV(\C^5 \setminus \C^2)\left\llbracket t\right\rrbracket$ with respect to just the $\dbar$ operator.

Since the exterior powers of the tangent bundle of $\C^5 \setminus \C^2$ are all trivial, the Dolbeault cohomology of polyvector fields will be the Dolbeault cohomology of the structure sheaf tensor an exterior algebra.   We can further decompose this (up to issues of completion) as a tensor product of the Dolbeault cohomology of $\C^2$ with that of $\C^3 \setminus 0$.
 
Let us recall how to describe Dolbeault cohomology of $\C^3 \setminus 0$.  We find that $H^0_{\dbar}(\C^3 \setminus 0)$ is the space $\Oo(\C^3)$ of holomorphic functions on $\C^3$.  Further, $H^i_{\dbar}(\C^3 \setminus 0) = 0$ if $i = 1,3$. Finally, we have a dense embedding
$$
z_1^{-1} z_2^{-1} z_3^{-1} \C[z_1^{-1}, z_2^{-1} z_3^{-1}] \into H^2_{\dbar}(\C^3 \setminus 0). 
$$
This is a module for the algebra $\C[z_i]$ of polynomials in the $z_i$, where we multiply in the evident way but set any monomial which contains a non-negative power of any $z_i$ to zero.

Alternatively, up to completion, we can describe $H^2_{\dbar}(\C^3 \setminus 0)$ as being the linear dual of the space of holomorphic top-forms on $\C^3$. The pairing is given by an integral over the unit $5$-sphere.  

In this description, the cohomology class $[F]$ of $F$ is
$$
[F] = z_1^{-1} z_2^{-1} z_3^{-1} \dpa{w_1} \dpa{w_2}
$$ 
(up to some factors of $\pi$). 

It follows that (up to completion)
$$
H^\ast( \PV(\C^5 \setminus \C^2), \dbar ) = H^\ast(\PV(\C^5)) \oplus H^2_{\dbar}(\C^3 \setminus 0) [w_i, \partial_{w_i}, \partial_{z_i}].  
$$
In a similar way,  we have a quasi-isomorphism (again up to completion) 
\begin{equation*}
H^\ast(\PV(\C^5 \setminus \C^2)\left\llbracket t\right\rrbracket, \dbar + t \partial ) \simeq H^\ast(\PV(\C^5)\left\llbracket t\right\rrbracket ) \oplus H^2_{\dbar}(\C^3 \setminus 0) [w_i, \partial_{w_i}, \partial_{z_i} ]\left\llbracket t\right\rrbracket \tag{$\dagger$} 
\end{equation*}
where on the second summand on the right hand side we have the differential $t\partial$ as usual. 

This describes the first page of the spectral sequence converging to the cohomology including the term $\{F,-\}$. On the next page, we simply introduce a differential given by the Schouten bracket with the cohomology class $[F]$ of $F$. This differential maps the terms in Dolbeault degree $0$ (which extend to holomorphic objects on $\C^5$) to those in Dolbeault degree $2$.   There are no further differentials in the spectral sequence, for degree reasons.

The copy of $\mf{psl}(3 \mid 3)$ we wrote down above lives inside $H^\ast(\PV(\C^5))$ and is in the kernel of the operator $\partial$.  We need to check it survives on all pages of the spectral sequence.  Since the spectral sequence degenerates after the second page, we need only verify that these operators are in the kernel of $\{[F],-\}$.  This is easy to verify explicitly: for example, 
\begin{align*} 
 \{[F], w_1 z_1   \} & = \{ z_1^{-1} z_2^{-1} z_3^{-1} \partial_{w_1} \partial_{w_2}, w_1 z_1\} \\
&= z_1 (z_1^{-1} z_2^{-1} z_3^{-1} ) \partial_{w_2} \\
&= 0 
\end{align*}
because the element $z_1^{-1} z_2^{-1} z_3^{-1}$ is in the kernel of the operators of multiplication by each $z_i$.

\end{proof}

\newcommand{\etalchar}[1]{$^{#1}$}

\end{document}